\documentclass[fleqn,11pt]{article}
\usepackage[margin=1in]{geometry}                
\usepackage{graphicx}
\usepackage{amssymb}
\usepackage{amsmath}
\usepackage{xfrac}
\usepackage{authblk}
\usepackage{setspace}
\usepackage{hyperref}
\usepackage{amsfonts}
\usepackage{pdflscape}
\usepackage{xcolor}
\usepackage{amsthm}

\hypersetup{colorlinks=true,linkbordercolor=red,linkcolor=blue,pdfborderstyle={/S/U/W 1},citecolor=blue}

\usepackage{dcolumn}
\usepackage{bm}
\usepackage{tikz}
\usetikzlibrary{shapes,decorations}

\usepackage{enumitem}

\usepackage{multirow,makecell}
\usepackage{array}

\newenvironment{myenumerate}
{ \begin{enumerate}
		\setlength{\itemsep}{0pt}
		\setlength{\parskip}{0pt}
		\setlength{\parsep}{0pt}       }
	{ \end{enumerate}                  } 

\usepackage[ruled,linesnumbered]{algorithm2e}

\usepackage{chemfig,chemformula}
\setchemfig{atom sep=20pt,+ sep left=1em,+ sep right=1em,+ vshift=0pt}

\newtheorem{Def}{Definition}
\newtheorem{Lemma}{Lemma}
\newtheorem{assume}{Assumption}

\newcommand\RR{\mathbb{R}}

\newcommand\EE{\mathbb{E}}

\newcommand\pcc{p_{\ch{CC}}}
\newcommand\pch{p_{\ch{CH}}}
\newcommand\phc{p_{\ch{HC}}}
\newcommand\phh{p_{\ch{HH}}}
\newcommand\Nc{N_{\ch{C}}}
\newcommand\Nh{N_{\ch{H}}}
\newcommand\Ncc{N_{\ch{CC}}}
\newcommand\Nch{N_{\ch{CH}}}
\newcommand\Nhh{N_{\ch{HH}}}
\newcommand\KHHE{ K_{\ch{HH}}^\text{eff}}
\newcommand\doubleC{\ch[bond-length=3pt, bond-offset=0pt]{C=C}}
\newcommand\KDoubleE{ K_{\doubleC}^\text{eff} }
\newcommand\bp{\bar{p}} 
\newcommand\bG{\bar{G}} 
\newcommand\tp{\tilde{p}}
\newcommand\tq{\tilde{q}} 
\newcommand\tG{\tilde{G}}

\begin{document}
	
	{\title{Cyclic random graph models predicting giant molecules in hydrocarbon pyrolysis}
		
		
		\author[1]{Perrin E. Ruth\thanks{pruth@umd.edu}}
		\author[2]{Vincent Dufour-D\'ecieux\thanks{vdufour@ethz.ch}}
		\author[3]{Christopher Moakler\thanks{cmoakler@gmail.com}}
		\author[1]{Maria Cameron\thanks{mariakc@umd.edu}}
		\affil[1]{\small{Department of Mathematics, University of Maryland, College Park, MD 20742, USA}}
		\affil[2]{\small{Energy \& Process Systems Engineering, ETH Zürich, Zürich 8092, Switzerland}}                               
		\affil[3]{\small{Applied Physics Laboratory, Johns Hopkins.}}
		
		\maketitle
		
		\begin{abstract}
			Hydrocarbon pyrolysis is a complex chemical reaction system at extreme temperature and pressure conditions involving large numbers of chemical reactions and chemical species. Only two kinds of atoms are involved: carbons and hydrogens. Its effective description and predictions for new settings are challenging due to the complexity of the system and the high computational cost of generating data by molecular dynamics simulations. On the other hand, the ensemble of molecules present at any moment and the carbon skeletons of these molecules can be viewed as random graphs. Therefore, an adequate random graph model can predict molecular composition at a low computational cost. We propose a random graph model featuring disjoint loops and assortativity correction and a method for learning input distributions from molecular dynamics data.  The model uses works of Karrer and Newman (2010) and Newman (2002) as building blocks. We demonstrate that the proposed model accurately predicts the size distribution for small molecules as well as the size distribution of the largest molecule in reaction systems at the pressure of 40.5 GPa, temperature range of 3200K--5000K, and H/C ratio range from 2.25 as in octane through 4 as in methane.
		\end{abstract}
		
		}


	\begin{figure}[!htbp]
		\centering
		\includegraphics{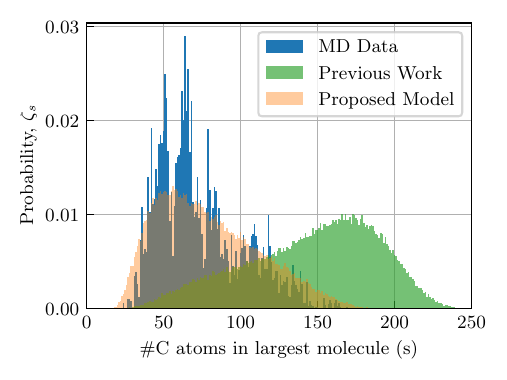}
		\caption{
			Predictions for the distribution of the number of carbon atoms in the largest molecule arising in ReaxFF MD simulations (blue) initialized to \ch{C4H10} at 3600K. The largest molecule corresponds to a giant connected component in random graph models. Previous work~\cite{dufour-decieuxPredicting2023} using the configuration model~\cite{molloyCritical1995,molloySize1998,newmanRandom2001} (green) overestimates the size of the largest molecule, and the proposed model 
			(orange) accurately predicts the size of the largest molecule.}
		\label{fig:largest_mol_ex}
	\end{figure}
	
	\section{Introduction}
	
	Hydrocarbon pyrolysis is a complex chemical reaction system of carbon and hydrogen occurring at extreme temperatures and pressures. It involves thousands of chemical reactions and tens of thousands of chemical species~\cite{yangLearning2017}. In particular, we consider conditions similar to the interiors of the ice giants Uranus and Neptune, where temperature is on the order of thousands of kelvins and pressure is on the order of tens of gigapascals. At sufficiently high pressures, carbon atoms separate from hydrogen atoms to form diamond~\cite{rossIce1981}. The conditions that lead to the formation of diamonds are studied in many works, e.g.~\cite{rossIce1981,ancilottoDissociation1997,krausFormation2017,chengThermodynamics2023}. At lower pressures, the pyrolysis of hydrocarbons results in a polymer state~\cite{hiraiPolymerization2009,krausFormation2017,ancilottoDissociation1997,spanuStability2011}. In this state, large molecules are still able to form and their network structure more closely resembles a tree. 
	
	It is challenging to study hydrocarbon pyrolysis with physical experiment due to extreme thermodynamic conditions and varying experimental setup.
	Experimental studies of hydrocarbon pyrolysis primarily use either diamond anvil cells~\cite{benedettiDissociation1999,hiraiPolymerization2009} or shock compression~\cite{krausFormation2017}. Unfortunately, these setups often disagree with each other, e.g. on which conditions are favorable for the formation of diamond. Thus, computational studies of hydrocarbon pyrolysis are needed to further understand the mechanisms that underlie hydrocarbon pyrolysis.
	
	Typically, computational studies of hydrocarbon pyrolysis use molecular dynamics (MD) simulations. The most accurate and versatile method for simulating reactive chemical systems is \textit{ab initio} MD -- see e.g.~\cite{ancilottoDissociation1997,spanuStability2011} for simulations of hydrocarbon pyrolysis with \textit{ab initio} MD. However, in our present setting involving thousands of atoms and nanosecond timescales, \textit{ab intio} MD simulations are prohibitively expensive.
	Classical MD force fields, in particular ReaxFF~\cite{vanduinReaxFF2001,chenowethReaxFF2008,srinivasanDevelopment2015,ashrafExtension2017}, present a computatationally affordable and commonly used alternative to \textit{ab initio} MD. In this work, our input data is generated using ReaxFF. It is worth mentioning, modelling of force fields is an active area of research~\cite{musilPhysicsInspired2021,unkeMachine2021,chengThermodynamics2023}.
	
	A cheaper, commonly used approach is to simulate hydrocarbon pyrolysis using kinetic Monte Carlo~\cite{gillespieGeneral1976,highamModeling2008} with reactions learned from MD data, see e.g.~\cite{yangDataDriven2019,chenTransferable2019}. Kinetic models of hydrocarbon pyrolysis are limited by combinatorial growth, requiring tens of thousands of reactions to model the formation of small molecules~\cite{chenTransferable2019}. Furthermore, kinetic models of hydrocarbon pyrolysis are unable to make accurate predictions about the size of the largest molecule. Recent work has shown that ``local'' reactions, i.e. reactions that only consider changes near the reaction site, reduce the combinatorial complexity of these kinetic models and make useful predictions for the size of the largest molecule~\cite{dufour-decieuxAtomicLevel2021,dufour-decieuxTemperature2022}. 
	
	MD simulations take on the order of days to weeks to reach nanosecond timescales, where simulation time is highly dependent on hardware. In contrast, the runtime of kinetic Monte Carlo simulations take on the order of minutes to hours. As with MD, kinetic Monte Carlo simulations must be sampled for each set of initial conditions (i.e. temperature, pressure, and initial composition). 
	
	In this work, as in the recent work by Dufour-D\'ecieux \textit{et al.}~\cite{dufour-decieuxPredicting2023}, we pursue a different approach, based on random graph models, to predict the molecule size distribution in hydrocarbon pyrolysis. 
	
	Random graph models are a fundamental tool for the study of complex systems such as social networks, infrastructure networks, and biological networks~\cite{ barabasiNetwork2013,newmanStructure2003,newmanNetworks2018}. Recent works by Torres-Knoop \textit{et al.}~\cite{torres-knoopModeling2018} and Schamboeck \textit{et al.}~\cite{schamboeckColoured2020} applied random graph theory to polymerization systems. Of particular interest to us is the \textit{configuration model}~\cite{molloyCritical1995,molloySize1998,newmanRandom2001} in which nodes have stubs sampled from a prescribed degree distribution and the stubs are randomly matched to form edges.
	
	The previous work~\cite{dufour-decieuxPredicting2023} demonstrated that the \emph{configuration model}~\cite{bollobasRandom2001,molloyCritical1995,newmanRandom2001} captured the small molecule size distribution of the carbon skeleton in hydrocarbon pyrolysis, where the carbon skeleton is the subgraph of a hydrocarbon network induced by carbon atoms. However, it notably overestimated the size of the largest molecule (Fig~\ref{fig:largest_mol_ex}) in systems with hydrogen/carbon (H/C) ratio $\le 2.5$ where the corresponding random graph had a giant component. 
	
	Our understanding is that the reason for this poor prediction is the property of the configuration model that the neighborhoods of nodes are tree-like~\cite{newmanRandom2001}. This property enables an analytical calculation of the component size distribution and the giant component size using a generating function approach~\cite{newmanRandom2001}. On the other hand, the tree-like property implies that small loops are rare. However, our analysis showed that molecules generated by ReaxFF MD simulations of hydrocarbon pyrolysis contain many small loops (Fig.~\ref{fig:Example_Mol}).
	\begin{figure}[!tb]
		\centering
		\includegraphics[width= \textwidth]{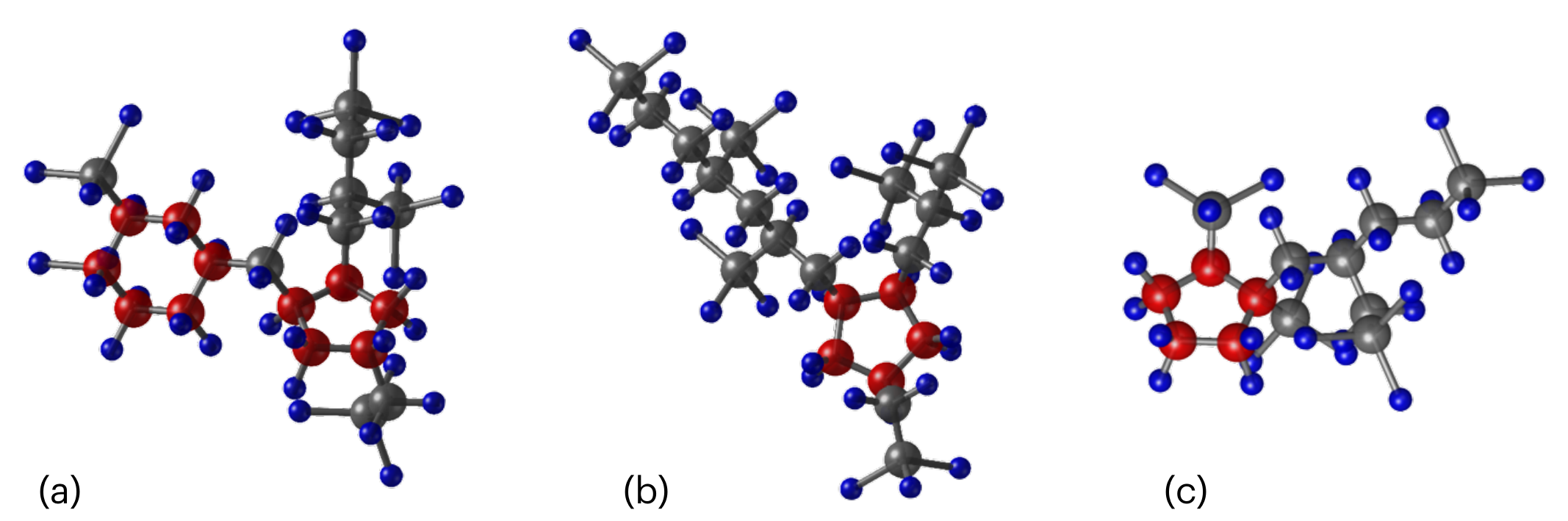}
		\caption{
			Three sample molecules from ReaxFF MD simulations that illustrate the structures we consider in this work. Carbon atoms in loops are shown in red and the remaining are shown in grey. Hydrogen atoms are blue. Molecule (a) has two disjoint carbon rings of length 5 and 6 separated by a single atom. Molecules (b) and (c) have a single carbon ring of length 5.
		}
		\label{fig:Example_Mol}
	\end{figure}

	\emph{The goal of this work is to design a random graph model that accurately describes molecule size distribution, including small molecule counts and the size distribution for the largest molecule, observed in MD simulation data.} Specifically, we aim to predict these distributions for any temperature and initial molecular composition within the span of data from~\cite{dufour-decieuxPredicting2023}, where the pressure is fixed at 40.5 GPa, the temperature ranges from 3200K to 5000K, and H/C ratio ranges from 2.25 as in \ch{C8H18} to 4 as in \ch{CH4}. In this work, we propose an all the way through scheme for predicting the molecule size distribution in hydrocarbon pyrolysis. That is, we model hydrocarbon pyrolysis using random graphs, and we obtain the resulting input parameters as functions of the H/C ratio and temperature using a local reaction model.
	
	The key component of this scheme is a random graph model, the \textit{Disjoint Loop Model} with \textit{Assortativity Correction}. This model accurately predicts the size distribution of the largest molecule as well as the small molecule size distribution. This model can be exercised in two ways: via sampling and via a generating function~\cite{flajoletAnalytic2009,wilfGeneratingfunctionology2005} approach. Sampling is effective when the number of carbons $\Nc$ is not too large, as in MD data, accounting for finite graph size. Generating functions are accurate in the limit $\Nc\to\infty$.
	
	The inputs for this model are the degree distribution, the loop rate per carbon atom, and the loop size distribution. We propose to obtain the degree distribution from a generating function with two parameters, the probability for hydrogen to bond to hydrogen rather than carbon, the probability for carbon to bond to 3 atoms rather than 4. We also develop a scheme for sampling the degree sequence. To obtain these parameters as well as the loop rate and loop size distribution we develop a sequence of local reaction models whose equilibrium constants are functions of this data. These constants are learned using Arrhenius fits to MD data. 
	
	Our scheme is depicted as a flowchart in Fig.~\ref{fig:model-flowchart}. A Python implementation is available on GitHub~\url{https://github.com/perrineruth/Disjoint-Loop-Model-Hydrocarbon-Pyrolysis}.
	
	\begin{figure}[!tbp]
		\centering
		\begin{tikzpicture}
			\small
			
			\node[anchor=south, minimum width = 2.75in, minimum height = 15pt, inner xsep = 0.025in, inner ysep = 1pt,fill=black!10,draw] (MD)
			at (0,2.95in + 30pt) {\textbf{MD Data}};
			
			{
				\node[anchor=south, minimum width = 2.75in, minimum height = 15pt, inner xsep = 0.025in, inner ysep = 1pt,fill=black!10,draw] (ReaxH)
				at (0,2.75in + 15pt) {\textbf{Reaction Model}};
				\node[anchor=south,text width = 2.65in, draw, inner xsep=0.05in, minimum height = 1.35in] (ReaxB)
				at (0,1.4in + 15pt) 
				{\raggedright
					1. Assemble local reactions for bond and \\[1pt]\phantom{1. }ring formation (Sec.~\ref{Sec:Parameter_Fit} and Appendix~\ref{app:Params})\\[1pt]
					2. For each reaction:\\[1pt]
					\phantom{2. }a. Estimate equilibrium constant \\[1pt]\phantom{2. a. }$K$ for all initial conditions\\[1pt]
					\phantom{2. }b. Fit $K$ to Arrhenius law $K=Ae^{-C/T}$ \\[1pt]\phantom{2. a. }(Table~\ref{tab:Arrh_Param})
				};
			}
			
			{
				\node[anchor=south, minimum width = 2.75in, minimum height = 15pt, inner xsep = 0.025in, inner ysep = 1pt,fill=black!10,draw] 
				(ParamH) at (0in,1in) {\textbf{Input parameters:}};
				\node[anchor = south, text width = 2.65in, draw, inner xsep = 0.05in,minimum height = 0.6in] 
				(ParamB) at (0in,0.4in) 
				{   \raggedright
					$\bullet$ $\{p_k\}$ = degree distribution\\[1pt]
					$\bullet$ $\lambda$ = loop rate per node\\[1pt]
					$\bullet$ $\{\phi_k\}$ = loop length distribution
				};
			}
			
			{
				\node[anchor=north, minimum width = 1.725in, minimum height = 15pt, inner xsep = 0.025in, inner ysep = 1pt,fill=black!10,draw] 
				(SampH) at (-0.88in,0in) 
				{\textbf{Sampling}};
				\node[anchor = north, text width =1.675in, draw, inner xsep = 0.025in, minimum height=1.1in] (SampB) at (-0.88in,0in-15pt) 
				{
					1. Sample Disjoint Loop
					\phantom{1. }Model, Algorithm~\ref{alg:DisLoop}\vspace{2pt}\\
					2. Add Assortativity\\
					\phantom{2. }Correction, Algorithm~\ref{alg:rewire}
				};
			}
			
			{
				\node[anchor=north, minimum width = 1.725in, minimum height = 15pt, inner xsep = 0.025in, inner ysep = 2pt,fill=black!10,draw] 
				(GenFH) at (0.88in,0in) 
				{\textbf{Generating Functions}};
				\node[anchor = north, text width =1.675in, draw, minimum height = 1.1in, inner xsep = 0.025in] 
				(GenFB) at (0.88in,0in-15pt) 
				{
					1. Solve for $H_j^E(x)$ \\[0pt]
					\phantom{1. }$j=0,1,2,3$, Eq.~\eqref{eq:BranchAssortRec},\\[0pt]
					\phantom{1. }using simple iteration\\[0pt]
					2. Compute $H(x)$, Eq.~\eqref{eq:HSolveAssort}\\[0pt]
					3. Extract coefficients $\{P_s\}$\\[0pt]
					\phantom{3. }of $H(x)$, Eq.~\eqref{eq:FFT}
				};
			}
			
			\node[anchor=north, minimum width = 2.25in, minimum height = 15pt, inner xsep = 0.025in, inner ysep = 2pt,fill=black!10,draw] 
			(Comp) at (0in,0in-15pt-1.45in) 
			{\textbf{Component Sizes}};
			
			\draw[thick,->] (MD) -- (ReaxH);
			\draw[thick] (MD) -- (-1.8in,2.95in + 37.5pt);
			\draw[thick] (-1.8in,2.95in+37.5pt) -- node[above,midway,rotate=90] () {time-averaging} (-1.8in,0.7in);
			\draw[thick,->] (-1.8in,0.7in) -- (ParamB);
			
			\draw[thick,->] (ReaxB) --
			node[anchor=east,align=right] () {Eqs.~\eqref{eq:phh_param_start}--\eqref{eq:lambdaphi_largeL_end},} 
			node[anchor=west,align=left] () {Sampling: Algorithm~\ref{alg:DegDist},\\
				Gen. Functions: Eq.~\eqref{eq:DegDist_fromParam}}
			(ParamH);
			
			\draw[thick,->] (ParamB) -- node[left] () {Small $N$} (SampH);
			\draw[thick,->] (ParamB) -- node[right] () {Large $N$} (GenFH);
			
			\draw[thick,->] (SampB) -- node[left] () {Depth-first search} (Comp);
			\draw[thick,->] (GenFB) -- node[anchor=west,align=right] () {\ \ Giant comp., Eq.~\eqref{eq:SfromH} \\
				Small comps., Eq.~\eqref{eq:PtoPi}} (Comp);
			
		\end{tikzpicture}
		\caption{
			Flowchart of the computational steps used to implement the Disjoint Loop Model with Assortativity Correction. The input parameters are estimated using a local equilibrium reaction model as a function of temperature and H/C ratio using Eqs.~\eqref{eq:phh_param_start}--\eqref{eq:lambdaphi_largeL_end}. The Disjoint Loop Model with Assortativity Correction can be studied using random graph sampling or generating functions. Random graph sampling accounts for finite graph size, and is useful for reproducing MD data. The degree sequence $\{k_i\}$ for random graph sampling is sampled using Algorithm~\ref{alg:DegDist}. The generating function approach is accurate in the limit of large graphs, and can be used to make predictions for large chemical systems. The degree distribution $\{p_k\}$ for the generating function approach is obtained via Eq.~\eqref{eq:DegDist_fromParam}.
		}
		\label{fig:model-flowchart}
	\end{figure}
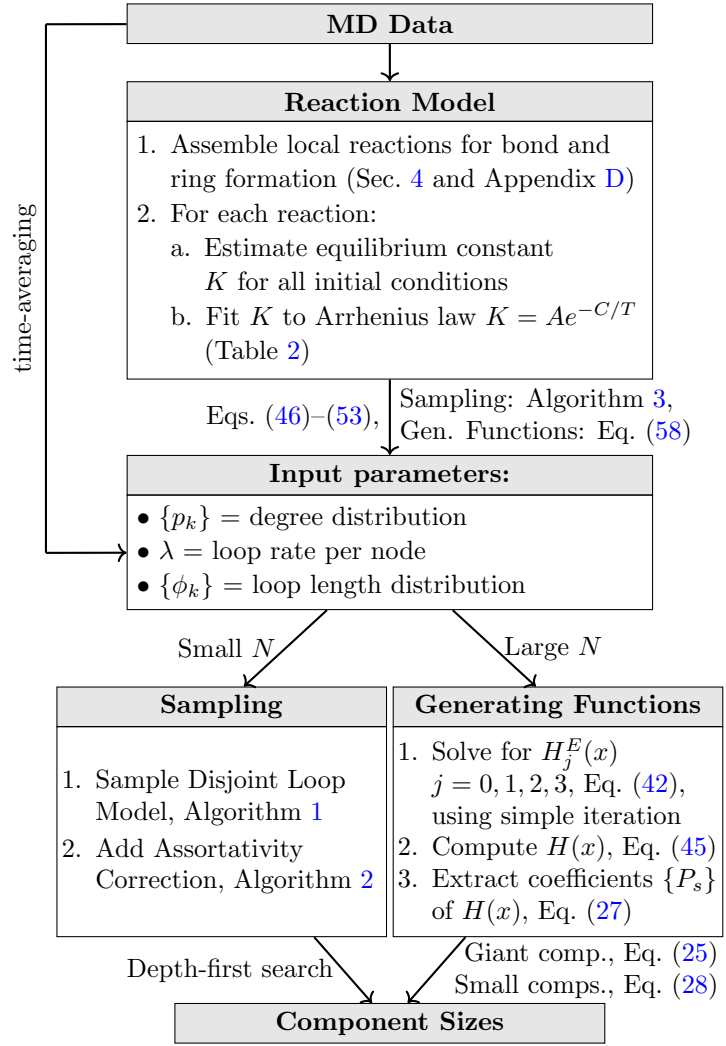
	
	The development of the Disjoint Loop Model with Assortativity Correction relies on several earlier works. The loops in MD data are present but sparse -- see Fig.~\ref{fig:Example_Mol}.  The sparsity and typically small loop sizes suggest a natural setting for Karrer's and Newman's model of random graphs with arbitrary subgraphs~\cite{karrerRandom2010}. The model in~\cite{karrerRandom2010} is tree-like at the level of subgraphs. We adopt this framework to model hydrocarbon pyrolysis, where the featured subgraphs are loops. This Disjoint Loop Model gives notably better predictions of the size of the largest molecule than the configuration model~\cite{dufour-decieuxPredicting2023}. We further improve its predictive power by the Assortativity Correction.
	The problem is that the Disjoint Loop Model has assortative mixing by degree, or edges are more likely to connect two nodes of similar degrees. This observation agrees with other studies of random graphs with subgraphs, e.g.~\cite{millerPercolation2009}. Assortative mixing is known to affect the size of the giant connected component as shown by Newman~\cite{newmanAssortative2002}. Therefore, we adapt the ideas from~\cite{newmanAssortative2002} to remove assortative mixing by degree from the Disjoint Loop Model.
	
	To our knowledge, the use of random graphs to model reactive chemical systems at equilibrium is quite new. We highlight some papers that also use this approach.
	Dobrijevic and Dutour~\cite{dobrijevicRandom2006} developed an inhomogeneous Erd\"os-Renyi-type random graph model to study hydrocarbon pyrolysis. It predicted a power law molecule size distribution of up to 6 carbons which matched data from Saturn's and Titan's atmospheres.
	As mentioned above, Refs.~\cite{torres-knoopModeling2018,schamboeckColoured2020} modelled polymer growth with random graphs.
	In this context, the presence of a giant connected component represents a polymer gel, and the Molloy-Reed criterion~\cite{molloyCritical1995} is used to compute the transition point for the formation of a gel.
	It is worth noting that the configuration model~\cite{torres-knoopModeling2018} underestimates the gel point and overestimates the size of the giant component.
	Importantly, this discrepancy is also explained by the presence of loops. 
	
	The rest of this paper is organized as follows. We provide the relevant background for random graph theory in Sec.~\ref{sec:Background_RGT}. In Sec.~\ref{sec:CycModel}, we introduce our Disjoint Loop Model with Assortativity Correction. In Sec.~\ref{Sec:Parameter_Fit}, we demonstrate how the parameters for our model are learned from MD data. 
	In Sec.~\ref{Sec:Case_Study}, we show that the proposed model predicts the size of the largest molecule in ReaxFF MD data. The results are summarized and perspectives of this approach are discussed in Sec.~\ref{sec:conclusion}.

	\section{Background: Extensions of the Configuration Model}
	\label{sec:Background_RGT}
	The configuration model~\cite{molloyCritical1995,molloySize1998,newmanRandom2001} is a classical random graphs model that outputs graphs with a prescribed degree sequence $\{k_i\}_{i=1}^N$, where $k_i$ is the degree of node $i$. The configuration model is sampled as follows:
	\begin{myenumerate}
		\item assign each node $i$, $k_i$ stubs representing half-edges,
		\item combine stubs pairwise uniformly at random,
		\item remove self-loops and multiple edges.
	\end{myenumerate}
	This procedure is depicted in Fig.~\ref{fig:ConfModelAlg}. Steps 1 and 2 output a multigraph (a graph with self-loops and multiple edges) where each pair of stubs is equally likely to be joined by an edge. Self-loops and multiple edges are removed in Step 3 to obtain a simple graph. In this sense, self-loops and multiple edges are errors, but they are rare. Specifically, the probability that a node participates in a self-loop or multiple edge is $O(N^{-1})$.
	
	\begin{figure}[!tbp]
		\centering
		\includegraphics{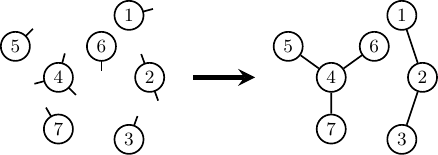}
		\caption{Example of sampling the configuration model when the degree sequence is $k_1=1$, $k_2=2$, $k_3=1$, $k_4=3$, $k_5=1$, $k_6=1$, and $k_7=1$. Each node $i$ is given $k_i$ stubs (left) then stubs are paired into edges uniformly at random (right). Here there are no self-loops or multiple edges that need to be removed.}
		\label{fig:ConfModelAlg}
	\end{figure}
	
	It is natural to study the configuration model in the limit of large graphs, i.e. number of nodes $N\to\infty$. In this limit, the behavior of the configuration model is a function of the degree distribution $\{p_k\}$, where $p_k$ is the fraction of nodes of degree $k$. 
	
	As the famous Erd\"os-R\'enyi model~\cite{solomonoffConnectivity1951,erdosEvolution1960}, the configuration model has a phase transition from the state where all components are small to the state where a giant component of size $O(N)$ is present~\cite{molloyCritical1995}. 
	
	Newman, Strogatz, and Watts~\cite{newmanRandom2001} proposed a generating function approach to compute the component size distribution in the limit $N\to\infty$. The generating function for the sequence $\{a_k\}$ is
	\begin{equation}
		G(x) = \sum_{k=0}^\infty a_k x^k.
	\end{equation}
	The generating function approach is implemented using the degree distribution $\{p_k\}$ and the excess degree distribution $\{q_k\}$ obtained from $\{p_k\}$ as follows. Let us pick a random edge, follow it in a random direction, and arrive at a node $i$. The number of other edges attached to node $i$ obey the excess degree distribution:
	\begin{equation}
		q_k = \frac{(k+1)p_{k+1}}{\sum_{k'}k'p_{k'}}.\label{eq:Back_ExDeg}
	\end{equation} 
	Indeed, the probability to reach a node with excess degree $k$ is proportional to the product of its degree $k+1$ and the probability for a node to have degree $k+1$ which is $p_{k+1}$.
	
	We refer the reader for more details on the configuration model and the generating function approach to the seminal paper by Newman, Strogatz, and Watts~\cite{newmanRandom2001}.

	In the next two subsections, we discuss two generalizations of the configuration model relevant to our work.
	
	\subsection{Motifs}
	\label{sub:motif}
	It has been observed as early as 1948 by Rapoport that small loops are rare in classical random graphs~\cite{rapoportCycle1948}. 
	Watts and Strogatz found that more triangles appear in real networks than in random graphs~\cite{wattsCollective1998}. For social networks, triangles are a consequence of edge transitivity, i.e. if A and B share a common friend C then A and B are more likely to be friends. Many models have been developed to account for the lack of triangles in random graphs, including~\cite{millerPercolation2009,newmanRandom2009,wattsCollective1998}. Subgraphs that appear more frequently in real networks than random graphs are often called network {motifs}, a term coined by Milo \textit{et al.}~\cite{miloNetwork2002}.
	
	The expected number of loops of length $\ell$ in the configuration model is given by
	\begin{equation}
		\frac{1}{2\ell}\left(\frac{\langle k^2\rangle -\langle k\rangle}{\langle k\rangle}\right)^\ell + O(N^{-1}),
		\label{eq:loop_count_Conf}
	\end{equation}
	where $\langle k\rangle$ and $\langle k^2\rangle$ are the expected degree and square degree, respectively. By Eq.~\eqref{eq:loop_count_Conf}, the expected number of loops of length $\ell$ is $O(1)$ and the probability a random node is in a loop of length $\ell$ is $O(N^{-1})$. Thus, if the second moment of the degree distribution is finite then small loops are rare in the configuration model. For completeness, we outline the computations that lead to Eq.~\eqref{eq:loop_count_Conf} in Appendix~\ref{app:Conf_Loop}. 
	
	The lack of loops in the configuration model is what makes it analytically tractable. Specifically, the configuration model is locally tree-like in the limit $N\to\infty$. This allows us to derive a self-consistency relationship for computing the small component size distribution -- see Fig. 3 in~\cite{newmanRandom2001}. 
	Thus, in the limit of large graphs, the size of the giant component and size distribution of small components can be obtained through a branching process~\cite{newmanRandom2001}. 
	
	Karrer and Newman~\cite{karrerRandom2010} developed a general framework for constructing random graphs that exhibit a finite collection of subgraphs with non-zero density, which we review here. To distinguish these subgraphs from arbitrary subgraphs, we will call these subgraphs ``motifs'' motivated by Milo \textit{et al.}~\cite{miloNetwork2002}. 
	The analysis of random graphs in~\cite{karrerRandom2010} relies on two sets of data: the motifs that they exhibit and the rate at which nodes participate in each motif. When the only motif is an edge, this model reduces to the configuration model. A more interesting example highlighted in Ref.~\cite{karrerRandom2010} 
	has edges and triangles as motifs. The probabilities of these motifs are described by the joint degree distribution $\{p_{s,t}\}$, where $p_{s,t}$ is the probability a node participates in $s$ edges (stubs) and $t$ triangles. {We emphasize, in this context edge motifs are the edges not in triangles.} The joint degree distribution is generated by a bivariate function
	\begin{equation}
		G(x,z) = \sum_{s,t} p_{s,t} x^s z^t.
	\end{equation}
	The dimensionality of the joint degree increases with the number of motifs that are introduced.
	
	Unlike the configuration model, the model in Ref.~\cite{karrerRandom2010} is not locally tree-like. This model remains tractable, however as it is locally tree-like at the level of the motifs. One corollary of this is that two motifs may not share two or more nodes. Due to the tree-like property at the level of motifs, the component size distribution in Ref.~\cite{karrerRandom2010} is computed using an upgraded version of the generating function approach in Ref.~\cite{newmanRandom2001}. In Section~\ref{sec:CycModel}, we restrict the model in Ref.~\cite{karrerRandom2010} to the case where the motifs are edges and simple loops.

	\subsection{Assortativity}
	\label{sub:Assort}
	A network is assortative if edges are more likely to connect nodes that are similar by a specified characteristic. For instance, our ReaxFF MD data exhibits disassortative mixing by atom type, i.e. carbon-hydrogen bonds are more favorable than carbon-carbon and hydrogen-hydrogen bonds. We construct an estimate for the degree distribution of the carbon skeleton relying on this bias.
	
	The other assortative mixing of interest is the one by degree.
	{One way to measure assortative mixing by degree is to use a Pearson correlation coefficient originally suggested by Newman~\cite{newmanAssortative2002}:
		\begin{equation}
			r = \frac{1}{\sigma_q^2}\sum_{j,k}jk(e^{\rm excess}_{jk} - q_j q_k).\label{eq:assort_coeff}
		\end{equation}
		Here $e^{\rm excess}_{jk}$ is the fraction of edges connecting nodes of excess degree $j$ and $k$, and $\{q_k\}$ is the excess degree distribution defined in Eq.~\eqref{eq:Back_ExDeg} with variance $\sigma_q^2$.} Random graphs that exhibit assortative mixing by degree ($r>0$) have a lower threshold for a giant connected component~\cite{newmanAssortative2002}. An intuition for this result is that networks with assortative mixing by degree will have a giant connected component composed of high-degree nodes that are more likely to be joined by an edge.
	
	In hydrocarbon networks, atom type can be determined by their degrees. Specifically, carbon atoms typically participate in 4 bonds and hydrogen atoms typically participate in a single bond. Thus, the degree assortativity coefficient $r$ on the global hydrocarbon network is an indirect measure for assortative mixing by atom type. Indeed, we find the global hydrocarbon network is disassortative by degree ($-0.8\le r\le -0.5$) -- listed in Table~\ref{tab:summary+giant}. This value of $r$ is much lower than in many social, technological, and biological networks~\cite[Table~I]{newmanAssortative2002}.
	
	For the configuration model, the degrees of nodes connected by an edge are independent, i.e. $e^{\rm excess}_{jk}=q_j q_k$. Thus, the configuration model lacks assortative mixing by degree ($r=0$). This makes the configuration model a useful baseline model for testing the impact of assortativity. 
	However, random graphs with motifs often exhibit assortative mixing by degree, e.g.~\cite{millerPercolation2009}.
	To see this in our setting, let edges and loops be the subgraphs of interest. To be in a loop, a node must be degree 2 or higher, and edge motifs are introduced in the same manner as in the configuration model. 
	Thus, edges in loops occur between pairs of high-degree nodes adding assortative mixing by degree ($r>0$).
	In our MD data, however, we find that the carbon skeleton either lacks assortative mixing by degree or is slightly disassortative -- see Table~\ref{tab:summary+giant} for the value $r$ for the carbon skeleton. In Sec.~\ref{subsec:rewire}, we show that the Disjoint Loop Model is assortative and correct for this using methods from Ref.~\cite{newmanAssortative2002}.
	
	In practice, the desired edge excess degree distribution $\{e^{\rm excess}_{jk}\}$ can be achieved by a rewiring algorithm proposed in Ref.~\cite{newmanAssortative2002}. Moreover, in~\cite{newmanAssortative2002}, the generating function approach was adapted for achieving the desired edge excess degree distribution $\{e^{\rm excess}_{jk}\}$ in the limit of large graphs.

	\section{The Proposed Model}
	\label{sec:CycModel}
	
	Now we present our model designed to describe and predict molecule size distributions in hydrocarbon pyrolysis. This model is developed in two stages. First, we setup the Disjoint Loop Model referring to its key feature. Then, we subject it to an Assortativity Correction to remove assortative mixing by degree. The resulting model is called the Disjoint Loop Model with Assortativity Correction.
	
	The Disjoint Loop Model can be studied using two methods. The first method is to sample from the random graph ensemble. The sampling algorithm (Algorithm~\ref{alg:DisLoop}) is introduced in Sec.~\ref{subsec:ModelConstruction}. The component sizes of a sample graph can be obtained via standard methods such as depth-first search (see e.g.~\cite{cormenIntroduction2022}). The second method analyzes the Disjoint Loop Model via generating functions in the limit $N\to\infty$. This approach is given in Sec.~\ref{subsec:Gen_DisLoop}, which is adapted from the work of Karrer and Newman~\cite{karrerRandom2010}. 
	In Sec.~\ref{subsec:rewire}, we introduce the Assortativity Correction, a method for removing assortative mixing by degree from the Disjoint Loop Model. Assortativity Correction modifies both the sampling scheme and the generating function formalism of the Disjoint Loop Model. These methods are adapted from the work~\cite{newmanAssortative2002}.
	
	
	

	The flowchart of the proposed model is depicted in Fig.~\ref{fig:model-flowchart}.
	
	\subsection{Disjoint Loop Model}
	\label{subsec:ModelConstruction}
	
	\subsubsection{Model assumptions}
	We start by listing the major assumptions that underlie the Disjoint Loop Model. These assumptions are restrictions we add to the general framework introduced by Karrer and Newman~\cite{karrerRandom2010}.
	As with any model, one must balance rigor with usability. Mathematically, our model must be simple enough to compute properties such as subgraph counts and component sizes. Practically, the parameters for this model must be attainable from data, e.g. large subgraphs are not necessarily observed in MD data and their counts may be noisy otherwise. To make our model tractable we introduce three assumptions. The first assumption is the subgraphs that interest us.
	\begin{assume}
The list of motifs consists of edges and simple loops of length $3,4,\hdots,L_{\max}$ {where $3 \le L_{\max}<\infty$.}
\label{asp:motifs}
\end{assume}
This choice of motifs has several benefits. From a percolation standpoint, loops contribute redundant edges. In other terms, an edge may be removed from a loop without breaking the loop into multiple components. From a chemistry standpoint, ring length impacts the stability of a molecule. A small ring, e.g. cyclopropane \ch{C3H6}, induces small, unstable bond angles. Larger rings, specifically six-membered rings, are much more stable and contribute to many interesting molecular structures such as diamond and graphene. From a practical standpoint, loops are easy to analyze. All nodes in a loop are equivalent under a rotation automorphism. Thus nodes in loops all play the same ``role'' and we do not need to consider the full role distribution introduced in~\cite{karrerRandom2010}. {The motifs that Assumption~\ref{asp:motifs} neglects are biconnected components that contain multiple loops. In this sense, Assumption~\ref{asp:motifs} is a first-order approximation to the motifs in our graphs.}

With our motifs fixed, we must determine how often each node participates in each motif, i.e. we must specify the fraction of nodes that participate in $k$ edge motifs, $\ell_3$ triangle motifs, $\ell_4$ square motifs, etc. This introduces a joint probability distribution of dimension $L_{\max}-1$. To avoid learning this distribution from MD data we make another restriction.
\begin{assume}
Let $u$ be a random node that participates in a loop motif. Then, the degree of $u$ and the length of the loop motif are independent.
\label{asp:LoopLenSep}
\end{assume}
This assumption decouples the length of a loop motif from the degrees of its participating nodes. This suggests that two probability distributions should be learned from MD data: (1) the probability a node is in $k$ edge motifs and $\ell$ loop motifs, $\tp_{k\ell}$, and (2) the probability a random loop motif is of length $k$, $\phi_k$. Here we use the tilde to distinguish the degree distribution split by edge and loop motifs $\{\tp_{k,\ell}\}$ from the classical degree distribution.   

As mentioned in Sec.~\ref{sub:motif}, motifs in the model introduced by Karrer and Newman rarely overlap by more than a single node. To enable the use of the tree-like hypothesis at the level of motifs in the generating sampling formalism we eliminate the possibility for a node to participate in multiple loops.
\begin{assume}
A node may only participate in a single loop motif, i.e. loop motifs are disjoint.
\label{asp:DisLoop}
\end{assume}
\noindent Assumption~\ref{asp:DisLoop} reduces the notational complexity of the Disjoint Loop Model and reduces the number of parameters that need to be learned. Namely, this assumption constrains the $\ell$ index of $\tp_{k,\ell}$ to binary values. 

While Assumption~\ref{asp:DisLoop} prohibits crystallization of carbons, it is adequate to accurately describe the MD data we have. To amend this, we would need to allow loops to share an edge in our random graph model. To make this change and keep the random graph model tree-like at the level of motifs, we would need to add new motifs to Assumption~\ref{asp:motifs}.
One could also introduce loops that are joined by a node. This does not break the locally tree-like assumption, but we do not consider it because it would introduce new parameters and we find this is rare in MD data. In particular, loops joined by a node introduce a node of degree 4 while only 1--3\% of C atoms participate in 4 carbon-carbon bonds, most of which do not participate in multiple loops.

\subsubsection{Sampling: Disjoint Loop Model}

\begin{algorithm}[t]
\KwIn{$N$ = number of nodes,\\
	\phantom{\textbf{Input: }}$k_i$ = degree of node $i=1,\hdots,N$ drawn from\\ \phantom{\textbf{Input: }}\quad degree distribution,\\
	\phantom{\textbf{Input: }}$L_j$ = length of loop $j=1,\hdots,N_L$ to be\\ \phantom{\textbf{Input: }}\quad included in the graph}
\KwOut{graph {$G=(V,E_E,E_L)$} where $E_E$ is the set of edge motifs and $E_L$ is the set of edges in loop motifs}
${ V}\gets\{1,\hdots,N\}$\;
Assign $k_i$ stubs to each node $i$\;
${S}\gets{}\emptyset$\tcp*{auxiliary set of nodes in loops}

\For{${n}=1$ to $\sum_{j=1}^{N_L} L_j$}{
	Add a random node ${ i}\in { V} \setminus {S}$ to ${S}$ with probability proportional to ${ w_{k_i}}{ = \binom{k_i}{2}}$\;
	Remove two stubs from node { $i$}\;
}
\tcp{Assign loop nodes to loops with shuffling}
$E_L\gets \emptyset$\;
\For{$j=1 \text{ to } {N_L}$}{
	Draw nodes $u_1,u_2,\hdots,u_{L_j}$ at random from {$S$}\;
	Remove $\{u_1,u_2,\hdots,u_{L_j}\}$ from  {$S$}\;
	Add $\{(u_1,u_2),(u_2,u_3),\hdots,(u_{L_j},u_1)\}$ to {$E_L$}\;
}
Attach remaining stubs together pairwise at random to make the set {$E_E$}\;
Remove self-loops and multiple edges (including those shared with $E_L$) from {$E_R$}\;
\caption{Disjoint Loop Model}
\label{alg:DisLoop}
\end{algorithm}

With Assumptions~\ref{asp:motifs},~\ref{asp:LoopLenSep} and~\ref{asp:DisLoop} 
fixed, we introduce a method for sampling random graphs. It is challenging to create a procedure for constructing random graphs using the node level distributions $\{\tp_{k,\ell}\}$ and $\{\phi_k\}$. Instead, it is easier to start by sampling loop motifs and assigning nodes to those loops. Our sampling scheme uses the following input parameters.
\begin{itemize}
\item The degree sequence $\{k_i\}_{i=1}^N$ where $k_i$ is the degree of node $i$ and $N$ is the number of nodes. These are obtained by $N$ independent and identically distributed (i.i.d.) samples from the degree distribution $\{p_k\}$ where $p_k$ is the probability a node has degree $k$.
\item Loop lengths $\{L_j\}_{j=1}^{N_L}$ where $L_j$ is the length of loop $j$ and $N_L$ is the number of loops. For hydrocarbon pyrolysis, loops are local structures that appear with some density. Thus, we assume the expected number of loops scales linearly with the number of nodes. We find that rings are sparse in MD data. In particular, we find that the number of loops $N_L$ closely follows a Poisson distribution, following the intuition that a binomial distribution with large $n$ and small $p$ is well-approximated by a Poisson distribution. Using these observations, let $\lambda$ be the loop rate per node and $\{\phi_k\}$ the loop length distribution. The loop lengths $L_1,\hdots,L_{N_L}$ are sampled as follows:
\begin{enumerate}
	\item let $N_L\sim\text{Poisson} (\lambda N)$ be the number of loops,
	\item let the loop lengths $L_1,\hdots,L_{N_L}$ be $N_L$ i.i.d. samples with PMF $\{\phi_k\}_{k=3}^{L_{\max}}$.
\end{enumerate}
\item Node weights $\{w_k\}$. These weights are used to sample nodes for loops, i.e. the probability a node of degree $k$ is drawn for a loop is proportional to $w_k$. We set $w_k=\binom{k}{2}$ indicating that there are $\binom{k}{2}$ ways of assigning two edges from a node of degree $k$ to serve as a ``corner'' of a loop. By a corner, we mean a node belonging to a loop with two stubs attached to it that participate in the loop.
\end{itemize}
Our sampling scheme is outlined in Algorithm~\ref{alg:DisLoop}. In short, nodes are sampled for loops using weighted sampling without replacement, i.e. nodes are chosen one at a time and at each step, the probability an unsampled node of degree $k$ is selected is proportional to $w_k$. Early draws in weighted sampling without replacement are biased towards large weights, so we shuffle nodes in loops before assigning them to particular loops. Finally, edge motifs are sampled using the configuration model.

\subsubsection{Loop inclusion probabilities}
The parameters we learn from MD data are $\{p_k\}$, $\lambda$, and $\{\phi_k\}$ matching the parameters used to sample random graphs. The generating function formalism for the Disjoint Loop Model uses the distributions $\{\tp_{k,\ell}\}$ and $\{\phi_k\}$. We now relate the distribution $\{\tp_{k,\ell}\}$ to the parameters used for sampling random graphs.

It turns out that this relationship is not easy to derive. We formulate it precisely in Lemma~\ref{lemma:inclusion_prob} below and prove it in Appendix~\ref{app:Inc_Prob}. Prior to this, we would like to give an intuition for relating $\{\tp_{k,\ell}\}$ to $\{p_k\}$, $\lambda$, and $\{\phi_k\}$. 

The nodes for loops can be sampled as follows:
\begin{myenumerate}
\item for each node $i$, add $w_{k_i}=\binom{k_i}{2}$ copies of $i$ to a list, where $k_i$ is the degree of node $i$,
\item sample nodes from this list uniformly at random until $\sum_{j=1}^{N_L} L_j$ unique nodes have been drawn.
\end{myenumerate}
Remark, that each node $i$ can be sampled any number of times from $0$ to $w_{k_i}$. As a result, there will be a fraction of the total node copies that are sampled equal to
\begin{equation*}
\hat{q} = \frac{\sum_{i=1}^N\text{\# sampled copies of node $i$}}{\sum_{i=1}^N w_{k_i}}.
\end{equation*}
A node of degree $k$ is not in a loop motif if none of its $w_k$ copies are drawn. In the limit of large graphs, this occurs with probability $(1-\hat{q})^{w_k}$. Therefore, the probability $f(k)$ a node of degree $k$ is in a loop motif is $1-q^{w_k}$, where $q=1-\hat{q}$, that satisfies the following equation
\begin{equation}
\sum_{k=2}^\infty p_k {(1-q^{w_k})} 
= \lambda \langle \phi_k\rangle, \quad \langle \phi_k\rangle=\sum_{k=3}^{L_{\max}}k\phi_k.  
\label{eq:q_identity}
\end{equation}
In other terms, Eq.~\eqref{eq:q_identity} states the fraction of nodes in loop motifs is equal to the expected number of loop motifs per node times the average loop length $\langle\phi_k\rangle$. There exists $q\in(0,1)$ that solves Eq.~\eqref{eq:q_identity} as long as there are enough nodes with degree $k\ge 2$ to construct the loop motifs, i.e.
\begin{equation}
\sum_{k=2}^\infty p_k > \lambda\langle \phi_k \rangle>0.
\label{eq:Loop_Feasibility}
\end{equation}
Furthermore, this solution $q$ is unique because the left hand side of Eq.~\eqref{eq:q_identity} is decreasing with respect to $q$. 


\begin{Lemma}
\label{lemma:inclusion_prob}
Assume $\{p_k\}$, $\lambda$, and $\{\phi_k\}$ satisfy Eq.~\eqref{eq:Loop_Feasibility} and let $q$ satisfy Eq.~\eqref{eq:q_identity}.
Then, the probability a node participates in $k$ edge motifs and $\ell$ loop motifs is
\begin{equation}
	\tp_{k,\ell} = \begin{cases}
		p_k (1-f(k)), &\ell=0\\
		p_{k+2}f(k+2), &\ell=1
	\end{cases}
	\label{eq:joint_deg_p}
\end{equation}
as $N\to\infty$, where
$f(k)$ is the asymptotic probability that a node of degree $k$ {is in a loop motif} given by
\begin{equation}
	f(k) 
	:={\begin{cases}0,& k = 0,1\\1-q^{w_k},& k\ge 2. \end{cases}}
	\label{eq:WSWRIncProb}
\end{equation}
	\end{Lemma}
	
	

	\subsubsection{Unintended loops}
	\label{ssub:err_loop}
	We end this subsection by justifying that small, unintended loops are rare in the Disjoint Loop Model. In particular, the probability a node participates in an unintended loop of length $\ell$ is $O(1/N)$. Alternatively, the expected number of extra loops of length $\ell$ is $O(1)$, and the expected fraction of nodes that are in unintended loops of length $\ell$ is $O(1)\ell/N = O(1/N)$.
	
	All extra loops must have at least one edge motif as loop motifs are disjoint.
	We find that it is useful to collapse loop motifs into a single node to justify that extra loops are rare. Under this perspective, extra loops are generalized nodes joined into a loop (including self-loops and multiple edges) by edge motifs. Edge motifs are constructed by connecting pairs of stubs at random, as in the configuration model. Using the results for the configuration model reviewed in Sec.~\ref{sub:motif}, we find that extra loops are rare if the degree distribution has finite second moment. Here the degree of a generalized node is the number of edge motifs it participates in. To verify this notion of degree distribution has finite second moments, it is sufficient to assure that $\{p_k\}$ and $\{\phi_k\}$ have finite second moments. We verify this in Appendix~\ref{app:Rare_Loops}. We also show how to compute the expected number of extra triangles in the Disjoint Loop Model in the limit $N\to\infty$ in Appendix~\ref{app:Triangles_DLM}.

\subsection{Generating function approach}
\label{subsec:Gen_DisLoop}



\subsubsection{A tree-like representation}
Here we compute the size distribution of small components and the size of the giant component for the Disjoint Loop Model in the limit of large graphs, $N\to\infty$. This is adapted from the generating function formalism in~\cite{karrerRandom2010}. The key observation introduced in~\cite{karrerRandom2010} is that nodes in a motif can be connected through a new node. For the Disjoint Loop Model this procedure is as follows. For each loop motif:
\begin{myenumerate}
\item remove all edges from the loop motif,
\item add a ``loop node'' to the graph,
\item add an edge between the loop node and every node in the loop motif.
\end{myenumerate}
This procedure is depicted in Fig.~\ref{fig:treelike_rep}. Ignoring loop nodes, this procedure results in a graph that has the same component sizes as the original graph. Furthermore, this representation of a graph is locally tree-like because loop motifs are broken and any additional loops are rare (as discussed in Sec.~\ref{ssub:err_loop}). Hence, we call this the ``tree-like representation'' of the Disjoint Loop Model. There are two types of nodes in the tree-like representation, which we call ``regular nodes'' and ``loop nodes''. Edges that connect pairs of regular nodes are called ``regular edges'' and edges that connect regular nodes to loop nodes are called ``loop edges''.

\begin{figure}[!tbp]
\centering
\includegraphics{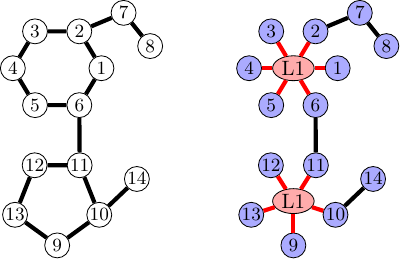}    
\caption{A sample graph (left) and its tree-like representation (right). This graph has a cycle of length 5 and 6 represented by the loop nodes labeled L1 and L2, respectively. In the tree-like representation regular edges (black) connect pairs of regular nodes (blue circles) and loop edges (red) connect regular nodes to loop nodes (red ellipses).}
\label{fig:treelike_rep}
\end{figure}

We remark that the tree-like representation differs from the ``bipartite graph representation'' given in~\cite{karrerRandom2010}. Recall, a bipartite graph is a graph with two types of nodes (here regular and loop nodes) such that edges only connect nodes of the opposite type. The tree-like representation is not bipartite because pairs of regular nodes may be connected by a regular edge. To make the graph bipartite, steps 1-3 would need to be repeated for edge motifs. We omit this step in this manuscript.

We will use the tree-like representation of a graph to compute the size distribution of small components as well as the size of the giant component. Specifically, the size of a connected component is the number of regular nodes it has. In the limit $N\to\infty$, traversing the tree-like representation is equivalent to a branching process. This branching process is depicted in Fig.~\ref{fig:Branching_Disjoint}. As outlined in the work~\cite{karrerRandom2010}, this branching process can be analyzed using multivariate generating functions, generalizing the work~\cite{newmanRandom2001}. For more information about generating functions see~\cite{wilfGeneratingfunctionology2005,flajoletAnalytic2009}. 

Ultimately, we want to compute the generating function $H(x)$ given by
\begin{equation}
H(x) = \sum\nolimits_s P_s x^s.
\end{equation}
where $P_s$ is the probability a random regular  node belongs to a connected component with $s$ regular nodes. Loop nodes are fictitious as they do not correspond to nodes in our original graph. In terms of our original graph, $P_s$ is the probability a random node belongs to a component to a connected component of size $s$. The small component size distribution and the size of the giant connected component may be computed from $\{P_s\}$ using the Newman-Strogatz-Watts generating function formalism~\cite{newmanRandom2001}. The generating functions used to compute $H(x)$ are summarized in Table~\ref{tab:Gen_Funcs}.

\begin{table*}[]
\centering
\def\w{9cm}
\begin{tabular}{|m{0.175\linewidth}<{}|m{0.175\linewidth}<{}|p{\w}|}
		\hline
		Generating Function\centering & Distribution\centering &  Description \centering\arraybackslash \\
		\hline\hline
		\multirow{1}{\linewidth}{\centering$G(x)$} & 
		\multirow{1}{\linewidth}{\centering$\{p_k\}$} & Classical Degree Distribution.  \\
		\hline
		\multirow{2}{\linewidth}{\centering$\tG(x,z)$} & 
		\multirow{2}{\linewidth}{\centering$\{\tp_{k,\ell}\}$} 
		& \multirow{2}{\w}{Degree distribution of a regular node in the tree-like representation split by regular and loop edges.} \\ & & \\
		\hline
		\multirow{2}{\linewidth}{\centering$\tG^E(x,z)$} & 
		\multirow{2}{\linewidth}{\centering$\{\tq_{k,\ell}^E\}$} & \multirow{2}{\w}{Excess degree distribution of a regular node reached by a regular edge.} \\ & & \\
		\hline
		\multirow{2}{\linewidth}{\centering$\tG^L(x)$\footnotemark[1]} & 
		\multirow{2}{\linewidth}{\centering$\{\tq_{k}^L\}$} & 
		\multirow{2}{\w}{Excess degree distribution of a regular node reached by a loop edge.} \\ & & \\
		\hline 
		\multirow{2}{\linewidth}{\centering$\Phi_0(x)$} & 
		\multirow{2}{\linewidth}{\centering$\{\phi_k\}$} & 
		\multirow{2}{\w}{Degree distribution of a random loop node or the loop length distribution.}\\ & &\\
		\hline 
		\multirow{2}{\linewidth}{\centering$\Phi_1(x)$} & 
		\multirow{2}{\linewidth}{\centering---} & 
		\multirow{2}{\w}{Excess degree distribution of a loop node reached by a random loop edge.}\\ & & \\
		\hline
		\multirow{3}{\linewidth}{\centering$H(x)$} & 
		\multirow{3}{\linewidth}{\centering$\{P_s\}$} & 
		\multirow{3}{\w}{Distribution for the number of regular nodes in a connected component reached by a random regular node.} \\ & & \\ & & \\
		\hline
		\multirow{3}{\linewidth}{\centering$H^E(x)$} & 
		\multirow{3}{\linewidth}{\centering---} & 
		\multirow{3}{\w}{Distribution for the number of regular nodes in a branch rooted by a regular node reached by a random regular edge.} \\ & & \\ & & \\
		\hline
		\multirow{3}{\linewidth}{\centering$H^E_j(x)$\footnotemark[2]} & \multirow{3}{\linewidth}{\centering---} & 
		\multirow{3}{\w}{Distribution for the number of regular nodes in a branch rooted by a regular node reached by a regular node of excess degree $j$.}\\ & &\\ & &  \\
		\hline
		\multirow{3}{\linewidth}{\centering$H^L(x)$} & 
		\multirow{3}{\linewidth}{\centering---} & 
		\multirow{3}{\w}{Distribution for the number of regular nodes in a branch rooted by a regular node reached by a random loop edge.} \\ & &\\ & &  \\
		\hline
\end{tabular}
\caption{List of generating functions used to analyze the tree-like representation of the Disjoint Loop Model and their associated probability distributions, `---' if undefined. A verbal description of each probability distribution is given as well. The letters $G$ and $\Phi$ denote degree distributions of regular and loop nodes, respectively. Tildes (e.g. $\tG(x,z)$) are used to denote a bivariate degree distribution that is split by regular and loop nodes. The letter $H$ denotes a generating function for the size distribution of a component or branch. Superscripts are used to denote edge type when a node is reached by a random edge.}
\footnotetext[1]{$\tG^L(x)=\tG^L(x,z)$ is univariate because a node reached by a loop edge can not participate in a second loop edge by Assumption~\ref{asp:DisLoop}.}
\footnotetext[2]{The generating functions $H_j^E(x)$ for $j=0,1,2,3$ replace $H^E(x)$ with assortativity correction. Implemented in Sec.~\ref{subsec:rewire}.}
\label{tab:Gen_Funcs}
\end{table*}

\subsubsection{Bivariate degree distribution for regular nodes}
We start by expressing the degree and excess degree distributions through generating functions. Let the degree distribution of our original graph, $\{p_k\}$, be generated by
\begin{equation}
G(x) = \sum_k p_k x^k.
\end{equation}
In the tree-like representation, it is useful to split neighbors by node type. The degree distribution of the tree-like representation is generated by a bivariate function
\begin{equation}
	\tG(x,z) = \sum_{k,\ell} \tp_{k,\ell}x^kz^\ell
	= \sum_k p_k\Big((1-f(k)) x^k + f(k) x^{k-2} z\Big),
\label{eq:DegDistGFGen}
\end{equation}
where $\tp_{k,\ell}$ is the probability given by Eq.~\eqref{eq:joint_deg_p} that a randomly drawn regular node has $k$ regular neighbors and $\ell$ loop neighbors. The tilde 
{emphasizes that } the degree distribution 
is split by regular and loop nodes for the tree-like representation. Note {the term $f(k)x^{k-2}z$ in Eq.~\eqref{eq:DegDistGFGen}. When} converting to the tree-like representation, one loop edge replaces two stubs of a vertex in a loop motif. Algebraically, $z$ replaces $x^2$ corresponding to the identity
\begin{equation}
G(x)=\tG(x,x^2).
\end{equation}
This can be verified by direct substitution into Eq.~\eqref{eq:DegDistGFGen}.

\subsubsection{Excess degree distributions for regular nodes}
Next, we generalize the notion of excess degree or the number of additional neighbors a node reached by an edge has. 
We use superscripts to denote the edge type that the node is reached by. Namely, let $\tq_{k,\ell}^E$ ($\tq_{k,\ell}^L$) be the probability that a node reached from a random regular (loop) edge has $k$ additional regular neighbors and $\ell$ additional loop neighbors. A node with $k$ regular neighbors contributes to $k$ stubs for regular edges. Thus, $\tq_{k,\ell}^E$ is proportional to $(k+1)\tp_{k+1,\ell}$, and the excess degree distribution of a regular node reached by a regular edge is:
\begin{equation}
\tq_{k,\ell}^E = \frac{(k+1)  \tp_{k+1,\ell}}{\sum_{k',\ell'}(k'+1)  \tp_{k'+1,\ell'}}.
\end{equation}
The {function generating the excess degree distribution of nodes reached by} regular edges is given by
\begin{equation}
	\tG^E(x,z) = \sum_{k,\ell} \tq_{k,\ell}^E x^k z^\ell
	=\frac{\sum_{k,\ell}(k+1) \tp_{k+1,\ell}x^kz^\ell}{\sum_{k',\ell'}(k'+1) \tp_{k'+1,\ell'}} = \frac{\tG_x(x,z)}{\tG_x(1,1)}.
\label{eq:GE_def}
\end{equation}
Similar arguments may be used to show that the excess degree distribution by loop edges, $\{\tq_{k,\ell}^L\}$, is generated by
\begin{align}
\tG^L(x,z) &= \frac{\sum_{k} p_{k}f(k) x^{k-2}}{\sum_{k'} p_{k'}f(k')}
= \frac{\tG_z(x,z)}{\tG_z(1,1)}.\label{eq:GL_def}
\end{align}
{Here and further, the subscripts $x$ and $z$ refer to the partial derivatives of the corresponding functions of $x$ and $z$.}
Note that $\tG^L(x,z)$ is constant with respect to $z$ because $\tG(x,z)$ is affine in $z$. Indeed, by Assumption~\ref{asp:DisLoop}, a node may only participate in a single loop, so a regular node reached by a loop edge cannot participate in an excess loop edge. Thus, we define the following shorthand $\tG^L(x)\equiv\tG^L(x,z)$ which generates $\tq_{k}^L\equiv\tq_{k,0}^L$.

\subsubsection{Degree distribution for loop nodes in the tree-like representation}
{Next, we} introduce a generating function for the degree distribution of loop nodes {in the tree-like representation:} 
\begin{equation}
\Phi_0(x) = \sum_{k=3}^{L_{\max}} \phi_k x^k,
\end{equation}
where $\phi_k$ is the probability a loop motif is of length $k$. $\Phi_0$ is univariate because all neighbors of a loop node are regular nodes. The excess degree of a loop node is generated by
\begin{align}
\Phi_1(x) &= \frac{\sum_{k=2}^{L_{\max}-1} k\phi_k x^k}{\sum_{k'=2}^{L_{\max}-1}k'\phi_{k'}}
= \frac{\Phi_0'(x)}{\Phi_0'(1)}.
\end{align}

We remark that the consistency equation, Eq.~\eqref{eq:q_identity}, may be written in terms of generating functions. The expected loop length is $\Phi_0'(1)$. The fraction of nodes that participate in a loop motif is given by
\begin{equation}
\tG_z(1,1) = \sum_{k} p_{k+2}f(k+2).
\end{equation}
Thus, Eq.~\eqref{eq:q_identity} is equivalent to
\begin{equation}
\tG_z(1,1) = \lambda \Phi_0'(1)
\label{eq:q_identity_GF}
\end{equation}
where $\tG_z(1,1)$ implicitly depends on $f$.

\begin{figure}[!tbp]
\centering
\includegraphics[]{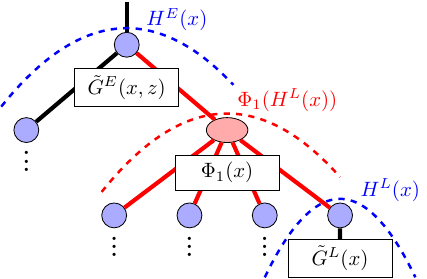}
\caption{Depiction of the branching process used to compute component sizes for the Disjoint Loop Model. The number of sub-branches a branch has, split by node type, follow the excess degree distributions generated by the functions $\tG^E(x,z)$, $\tG^L(x)$, and $\Phi_1(x)$. The nodes rooting sub-branches are also reached by a random edge. The size of sub-branches, in terms of regular nodes, follow the same distributions (generated by $H^E(x)$ and $H^L(x)$) as their parent branches (memoryless). Thus, the generating functions $H^E(x)$ and $H^L(x)$ are computed recursively.
}
\label{fig:Branching_Disjoint}
\end{figure}

\subsubsection{Branch size distributions}
The generating functions for the generalized degree distributions of the tree-like representation can be used to compute the size distribution of branches. This procedure is depicted in Fig.~\ref{fig:Branching_Disjoint}. Let $H^E(x)$ ($H^L(x)$) generate the number of regular nodes in a branch whose root is a regular node reached by a random regular (loop) edge. 
{The branch size distribution can be computed using the multiplication property of generating functions, i.e. if $X$ and $Y$ are independent random variables generated by $\mathcal{F}(x)$ and $\mathcal{G}(x)$ respectively, then $X+Y$ is generated by $\mathcal{F}(x)\mathcal{G}(x)$.}
The number of regular nodes in $k$ independent branches rooted by a regular node reached by a random regular edge is generated by $[H^E(x)]^k$. By conditioning on the excess degree of the root node, $H^L(x)$ is a function of $H^E(x)$
\begin{equation}
H^L(x) = x\sum_{k} \tq_{k}^L [H^E(x)]^k = x\tG^L(H^E(x))
\end{equation}
where the leading $x$ corresponds to the root node of the branch.
By the same conditioning argument, the number of regular nodes in a branch rooted by a loop node reached by a random loop edge is generated by
\begin{equation}
\Phi_1(H^L(x)) = \Phi_1(x\tG^L(H^E(x))).
\end{equation}
The leading $x$ is missing because the root is an artificially added loop node. 
By conditioning on the excess degree of a regular node reached by a regular edge, we get the following recurrence
\begin{equation}
	H^E(x) = x\sum_{k,\ell} \tq_{k,\ell}^E [H^E(x)]^k \big[\Phi_1(x\tG^L(H^E(x)))\big]^\ell
	= x\tG^E(H^E(x),\Phi_1(x\tG^L(H^E(x)))).
\label{eq:H_rec}
\end{equation}
The generating functions $\tG^E(x,z)$, $\tG^L(x)$, and $\Phi_1(x)$ may be computed directly from the input distributions $\{\tp_{k,\ell}\}$ and $\{\phi_k\}$.
With these distributions, we may compute $H^E(x)$ by solving Eq.~\eqref{eq:H_rec}. Without Assumption~\ref{asp:DisLoop}, $H^E(x)$ and $H^L(x)$ would have to be computed using a joint recurrence.

\subsubsection{Small component size distribution and the giant component size}
Finally, by the same conditioning arguments, the number of regular nodes in a connected component reached by a random regular node is generated by
\begin{equation}
H(x) = x\tG(H^E(x),\Phi_1(x\tG^L(H^E(x)))).
\label{eq:H_solve}
\end{equation}
Specifically, $H(x)$ is computed directly from $H^E(x)$ by Eq.~\eqref{eq:H_solve}. The fraction of nodes that are in the giant connected component is given by
\begin{equation}
\label{eq:SfromH}
S=1-\sum\nolimits_s P_s = 1-H(1).
\end{equation}
If $H(1)=1$ then all components are small.

The small component size distribution is computed using the coefficients of $H(x)$. On the unit circle $x=e^{i\theta}$, $H(x)$ is an absolutely convergent Fourier series
\begin{equation}
H(e^{i\theta}) = \sum\nolimits_s P_s e^{is\theta}.
\end{equation}
{Thus, the coefficients of $H(x)$ are computed using the Fast Fourier Transform. Explicitly, the Fast Fourier Transform of $\{H(e^{i\theta_j})\}_{j=0}^{2^J-1}$ is given by
\begin{equation}
	\label{eq:FFT}
	{P_s = \frac{1}{2^J}\sum_{j = 0}^{2^J-1}H(e^{i\theta_j})e^{-i\theta_js},\quad\theta_j = \frac{2\pi j}{2^J}.}
\end{equation}
Note that only the first $2^{J-1}$ values of $P_s$ should be used due to the discrete nature of the Fast Fourier Transform. We evaluate $H(x)$ by solving Eq.~\eqref{eq:H_rec} for $H^E(x)$ and substituting $H^E(x)$ into Eq.~\eqref{eq:H_solve}.


Recall, $P_s$ is the probability a random regular node belongs to a connected component with $s$ regular nodes. Note, a connected component of size $s$ has $s$ nodes that can be randomly selected, so the small component size distribution is obtained by accounting for multiplicity:}
\begin{equation}
\pi_s = \frac{P_s/s}{\sum_{s'}P_{s'}/s'}.
\label{eq:PtoPi}
\end{equation}
In the tree-like representation, $\pi_s$ is the probability a randomly selected connected component has $s$ regular nodes. Alternatively, $\pi_s$ is the probability a random connected component in the original graph has $s$ nodes.

\subsection{Assortativity Correction}
\label{subsec:rewire}
As observed by Miller~\cite{millerPercolation2009}, the addition of triangles into a random graph model results in assortative mixing by degree. We find the same result is true for the Disjoint Loop Model (Algorithm~\ref{alg:DisLoop}). {We do this by computing the degree assortativity coefficient $r$ and showing that it is non-negative. To address this, we propose a method for removing assortative mixing from the Disjoint Loop Model, i.e. the Assortativity Correction, which modifies both the sampling and generating function approaches.}

\subsubsection{The Pearson correlation coefficient via generating functions}
To characterize assortativity, we compute  {the Pearson correlation coefficient $r$ as} 
\begin{equation}
r=\frac{1}{\sigma_q^2} \sum_{j,k} jk(e_{jk}-q_{j-1} q_{k-1}).
\label{eq:assort_shift}
\end{equation}
Here, $e_{jk}$, is the fraction of all edges that connect nodes of total degree $j$ and $k$, and $q_k$ is the probability a node reached by a random edge from a node of excess degree $k$ given by
\begin{equation}
q_{k} = \frac{(k+1)p_{k+1}}{\sum_{k'}k'p_{k'}}.
\end{equation}
$\sigma_q^2$ is the variance of the excess degree distribution.

Note that $e_{jk} = e^{\rm excess}_{j-1,k-1}$. Using simple algebra, one can verify that  Eq.~\eqref{eq:assort_shift} is equivalent to the degree assortativity coefficient defined in~\cite{newmanAssortative2002} which is given by Eq.~\eqref{eq:assort_coeff}. We prefer using $e_{jk}$ rather than $ e^{\rm excess}_{j-1,k-1}$ because the excess degree depends on the type of motif a node is reached by.

Let $e_{jk}^E$ and $e_{jk}^L$ be the fraction of respective edges in edge and loop motifs in the original graph that connect nodes of degree $j$ and $k$.
These satisfy
\begin{equation}
e_{jk} = \frac{\tG_x(1,1)}{G'(1)} e_{jk}^E + \frac{2\tG_z(1,1)}{G'(1)} e_{jk}^L.
\label{eq:DisLoopEdgeRel}
\end{equation}
The ratios $\tG_x(1,1)/G'(1)$ and $\tG_z(1,1)/G'(1)$ are the fractions of edges in edge and loop motifs respectively. Using Eqs. \eqref{eq:WSWRIncProb} and \eqref{eq:DegDistGFGen} to compute the partial derivatives $\tG_x(1,1)$ and $\tG_z(1,1)$, we find
\begin{equation}
\frac{\tG_x(1,1)}{G'(1)} = \frac{\sum_k (k-2f(k))p_k}{\sum_k k p_k} = 1- \frac{2\tG_z(1,1)}{G'(1)}.
\label{eq:gen_d_rel}
\end{equation}
Eq.~\eqref{eq:gen_d_rel} can be verified by applying the chain rule to $G(x)=\tG(x,x^2)$.

Substituting Eq.~\eqref{eq:DisLoopEdgeRel} into Eq.~\eqref{eq:assort_shift} gives the following expression of the Pearson correlation coefficient:
\begin{equation}\begin{split}
	r = \frac{1}{\sigma_q^2}\sum_{j,k}jk\bigg[\frac{\tG_x(1,1)}{G'(1)}e_{jk}^E &+ \frac{2\tG_z(1,1)}{G'(1)}e_{jk}^L \\[-10pt]&\quad\quad\quad- q_{j-1} q_{k-1}\bigg].
\end{split}
\label{eq:assort_DisLoop_a}
\end{equation}

We now compute the values $e_{jk}^E$ and $e_{jk}^L$ in terms of excess degree distributions. Since the excess degrees of nodes connected by a regular edge are independent, the fraction of regular edges that connect nodes of degrees $j$ and $k$ is
\begin{equation}
e_{jk}^E = q_{j-1}^E q_{k-1}^E.
\label{eq:ex_reg_edge}
\end{equation}
Here 
\begin{equation}
q_k^E = \frac{(k+1-2f(k+1))p_{k+1}}{\sum_j (j-2f(j))p_j},\quad k=0,1,2,3,
\label{eq:ex_reg_deg}
\end{equation}
is the probability a node reached by a random regular edge has excess degree $k$. This should not be confused with the excess degree distribution $\{\tq_{k,\ell}^E\}$ which is split by motif type. Specifically, $q_k^E = \tq_{k,0}^E+\tq_{k-2,1}^E$ which is generated by $\tG^E(x,x^2)$. 

The fraction of loop edges that connect nodes of degree $j$ and $k$ is
\begin{equation}
e_{jk}^L = q_{j-2}^L q_{k-2}^L
\label{eq:ex_loop_edge}
\end{equation}
where 
\begin{equation}
q_k^L = \frac{f(k+2)p_{k+2}}{\sum_j f(j)p_j}, \quad k=0,1,2
\label{eq:ex_loop_deg}
\end{equation}
is the fraction of nodes reached by a loop motif that participate in $k$ excess edge motifs. Here $q_k^L=\tq_k^L$ is generated by $\tG^L(x)$. By conditioning on edge type the probability $q_k$ can be computed in terms of $q_{k}^E$ and $q_{k-1}^L$
\begin{equation}
q_k = \frac{\tG_x(1,1)}{G'(1)} q_k^E + \frac{2\tG_z(1,1)}{G'(1)} q_{k-1}^L 
\label{eq:q_equivalence}
\end{equation}
where $q_{-1}^L=0$. Alternatively, Eq.~\eqref{eq:q_equivalence} can be derived by summing Eq.~\eqref{eq:DisLoopEdgeRel} with respect to $j$. To simplify notation, define the moments
\begin{equation}
\begin{split}
	\mu_E &= \sum_k k q_{k-1}^E,  \quad\quad \mu_L = \sum_k k q_{k-2}^L, \quad\quad
	\mu = \sum_k kq_{k-1} = \frac{\tG_x(1,1)}{G'(1)}\mu_E + \frac{2\tG_z(1,1)}{G'(1)}\mu_L.
	\label{eq:ex_mean_id}
\end{split}
\end{equation}
If we represent the values $e_{jk}^E$ and $e_{jk}^L$ in Eq.~\eqref{eq:assort_DisLoop_a} by their excess degree distributions and replace the moments by their corresponding variables we find
\begin{equation}
	r = \frac{1}{\sigma_q^2}\left[\frac{\tG_x(1,1)}{G'(1)}\mu_E^2 + \frac{2\tG_z(1,1)}{G'(1)}\mu_L^2 - \mu^2\right] = \frac{2\tG_x(1,1)\tG_z(1,1)(\mu_E-\mu_L)^2}{\big(G'(1)\sigma_q\big)^2}.
\label{eq:assort_DisLoop_b}
\end{equation}
The second line in Eq.~\eqref{eq:assort_DisLoop_b} requires some tedious algebra. In short, we replace $\mu$ by $\mu_E$ and $\mu_L$ using Eq.~\eqref{eq:ex_mean_id} and simplify using Eq.~\eqref{eq:gen_d_rel}.

All terms in Eq.~\eqref{eq:assort_DisLoop_b} are non-negative verifying that the assortativity coefficient of the Disjoint Loop Model is non-negative, which we also see in graphs generated by Algorithm~\ref{alg:DisLoop}. In particular, if the expected degree of a node reached by an edge motif ($\mu_E$) does not match the expected degree of a node reached by a loop motif ($\mu_L$) then the Disjoint Loop Model is assortative ($r>0$).

\subsubsection{Assortativity Correction: Sampling}

\begin{algorithm}[t]
\KwData{Graph $G$ from Algorithm~\ref{alg:DisLoop}}
\For{$x = 1:N_{\text{Resample}}$}{
	Let $(u_1,v_1)$ and $(u_2,v_2)$ be regular edges whose nodes have degrees $i_1,j_1$ and $i_2,j_2$, respectively\;
	Let $U$ be a uniform random variable on $[0,1]$\;
	\If{$U \le (\hat{e}_{i_1 i_2} \hat{e}_{j_1 j_2})/(\hat{e}_{i_1 j_1} \hat{e}_{i_2 j_2})$}{
		Remove $(u_1,v_1)$ and $(u_2,v_2)$ from $G$\;
		Add $(u_1,u_2)$ and $(v_1,v_2)$ to $G$\;
	}
}
\caption{Assortativity Correction}
\label{alg:rewire}
\end{algorithm}

We add disassortative mixing to regular edges to remove assortative mixing from the Disjoint Loop Model. 
{Specifically,} we rewire regular edges so that the fraction of all edges that connect nodes of degree $j$ and $k$ is $e_{jk}=q_{j-1} q_{k-1}$. {The loop motifs are preserved.} 
The desired fraction of regular edges that connect nodes of excess degree $j$ and $k$ is given by Eq.~\eqref{eq:DisLoopEdgeRel}, namely
\begin{equation}
\hat{e}_{jk}^E = \frac{G'(1)}{\tG_x(1,1)}q_{j-1} q_{k-1} - { 2}\frac{\tG_z(1,1)}{\tG_x(1,1)} q_{j-2}^L q_{k-2}^L.
\label{eq:des_reg_edge}
\end{equation}
Given the values $\hat{e}_{jk}^E$, we rewire regular edges according to Algorithm~\ref{alg:rewire} adapted from~\cite{newmanAssortative2002}. The rewiring step in Algorithm~\ref{alg:rewire} preserves node degrees, it is ergodic, and the fraction of regular edges equilibriates to $\hat{e}_{jk}$ by a detailed balance argument. This procedure is feasible when the desired edge counts are non-negative, $\hat{e}_{jk}>0$. 

Note that when we run both Algorithms~\ref{alg:DisLoop} and~\ref{alg:rewire} we remove self-loops and multiple edges only after the rewiring is done.

\subsubsection{Assortativity Correction: Generating function formalism}
{The Assortativity Correction needs to be done also in the formulas resulting from the generating function approach described in Section \ref{subsec:Gen_DisLoop}.  In this Section, we derive a replacement for the generating function $H(x)$ for the small component size distribution $\{P_s\}$, Eq. \eqref{eq:H_solve}, using an approach adapted from ~\cite{newmanAssortative2002}.  }


Let $H^E_j(x)$ generate the size distribution of a branch whose root is reached by a regular edge from a node of degree $j=1,2,3,4$. The functions $H_j^E$ satisfy the following nonlinear system of algebraic equations:
\begin{equation}
	H^E_j(x) = \sum_{k} \frac{\hat{e}_{jk}}{\sum_{k'}\hat{e}_{jk'}}\Big((1-f_1(k))[H^E_k(x)]^{k-1} + f_1(k)[H^E_k(x)]^{k-3}\Phi_1(H^L(x)) \Big).
\label{eq:BranchAssortRec}
\end{equation}
The ratio ${\hat{e}_{jk}}/{\sum_{k'}\hat{e}_{jk'}}$ is the probability the root has degree $k$ and 
\begin{equation}
f_1(k) = \frac{(k-2)\tp_{k-2,1}}{k\tp_{k,0}+(k-2)\tp_{k-2,1}} 
= \frac{(k-2)f(k)}{k-2f(k)}
\end{equation}
is the probability a node of degree $k$ reached by a random regular edge is in a loop motif. 

The size distribution of a branch whose root is a regular node reached by a loop edge is generated by
\begin{equation}
H^L(x) = \sum_k q_k^L [H_{k+2}^E(x)]^{k}.
\label{eq:HL_assort}
\end{equation}
Note $H^L(x)$ is simply a function of the generating functions $H_j^E(x)$, which is a direct consequence of Assumption~\ref{asp:DisLoop}. Recall that $P_s$ is the probability a random node belongs to a component of size $s$. The values $\{P_s\}$ are generated by
\begin{equation}
	H(x) = xp_0 + x\sum_k p_k\Big((1-f(k))[H_{k}^E(x)]^k + f(k)[H_{k}^E(x)]^{k-2}\Phi_1(H_L(x))\Big).
\label{eq:HSolveAssort}
\end{equation}
Once $H(x)$ is computed, the fraction of nodes in the giant component and the small component size distribution can be computed using the methods discussed in Sec.~\ref{subsec:Gen_DisLoop}. To evaluate $H(x)$, we first solve     Eqs.~\eqref{eq:BranchAssortRec} and~\eqref{eq:HL_assort}  by simple iteration and input these values into Eq.~\eqref{eq:HSolveAssort}.

\section{Estimation of the Model Parameters}
\label{Sec:Parameter_Fit}


%

In this section, we discuss how the parameters for the Disjoint Loop Model with Assortativity Correction are obtained from MD data.
Specifically, we learn the following parameters used to model the carbon skeleton: (1) the degree distribution of the carbon skeleton, $\{p_k\}$, (2) the loop rate per carbon atom, $\lambda$, and (3) the loop length distribution, $\{\phi_k\}$. These parameters are a function of temperature and H/C ratio. The pressure is fixed. 

\subsection{Summary}

\begin{table}[]
\centering
\begin{tabular}{|>{\centering\arraybackslash}p{1.275in}| >{\centering\arraybackslash}p{.9in} | >{\centering\arraybackslash}p{.9in}|}
	\hline
	Equilibrium Constant & $A$ & $C$\\
	\hline\hline
	$\KHHE$ & $4.4056$ & 18704 \\
	$\KDoubleE$ & $6.46\times10^{-4}$ & 30526 \\
	$K_3^\text{eff}$ & $1.2390\times10^{-1}$ & 5451.1 \\
	$K_4^\text{eff}$ & $1.5152\times 10^{-4}$ & -28110\\
	$K_5^\text{eff}$ & $6.7734\times10^{-5}$ & -36238\\
	$K_6^\text{eff}$ & 4.8330 & 18949\\
	$K_7^\text{eff}$ & 5.2197 & 22168\\
	$K_8^\text{eff}$ & 4.4314 & 20251\\
	$K_L^\text{eff}$ & 4.4209 & 18254\\
	\hline
\end{tabular}
\caption{Arrhenius fits for all equilibrium constants $K^\text{eff}=Ae^{-C/T}$. Coefficients $A$ and $C$ are computed using linear least squares for the equation $\ln K^\text{eff} = \ln A - C/T$. For $\KHHE$ and $\KDoubleE$ Arrhenius laws are fit using \ch{C4H10} data. The remaining Arrhenius laws are fit using \ch{C4H10} data with temperature $T\ge3400K$ because some large loop lengths were not observed at low temperatures. The associated reactions are presented in Sec.~\ref{sub:deg_param_fit} and Appendix~\ref{app:Params_Loop}}
\label{tab:Arrh_Param}
\end{table}

We learn the input parameters for our model from a series of local equilibrium reactions. 
The procedure for this is as follows. For each parameter, we {choose an appropriate} local reaction, i.e. a reaction {that describes changes only near the reaction site}. An effective equilibrium constant is then fit to an Arrhenius law $K^\text{eff}=Ae^{-C/T}$, where $T$ is temperature and $K^\text{eff}$ is proportional to the equilibrium constant $K$. The constant of proportionality is either a combinatorial factor 
(e.g. $\KHHE=4K_{\rm HH}$) or a chemical parameter that does not depend on temperature or H/C ratio (e.g. $\KDoubleE=c_P K_{\ch{C=C}}$ where $c_P$ is a function of pressure). These effective equilibrium constants are listed in Table~\ref{tab:Arrh_Param}.

The input parameters 
for the degree distribution $\{p_k\}$ are $\bp_3$, $\phh$, and $\pch$ where $\bp_3$ is the probability a carbon is bonded to 3 atoms, $\phh$ is the probability a hydrogen bonds to another hydrogen, and $\pch$ is the probability for a bond from a carbon to attach to hydrogen. 
These input parameters are obtained from $\Nc$, $\Nh$, the equilibrium constants $\KHHE$ and $\KDoubleE$ from the following relationships: 
\begin{align}
\phh &= 1 - \frac{4\frac{\Nc}{\Nh}{+}1 - \sqrt{\left(4\frac{\Nc}{\Nh}{-}1\right)^2+16\KHHE\frac{\Nc}{\Nh}}}{2(1-\KHHE)},\label{eq:phh_param_start}\\
\pch &= \frac{1}{3\bp_3 + 4\bp_4}\frac{\Nh}{\Nc} (1-\phh),\\
\bp_3 &= \frac{\KDoubleE}{\KDoubleE+\tfrac{\phh}{(12.011\Nc+\Nh)T}}\label{eq:p3_param_end}
\end{align}
{where $12.011$ is the atomic mass of carbon.} Here $\KHHE$ and $\KDoubleE$ are associated with reactions of forming \ch{H2} and double bonds respectively.

The loop rate $\lambda$ and the loop size distribution $\{\phi_k\}$ are obtained from appropriate equilibrium constants associated with either ring formation or expansion reactions.
Counting loops or rings in a graph is non-trivial. Many loops should not be counted because they are composed of smaller loops. In this work, we only consider loops (cycles) that are part of a ``minimum cycle basis''. Minimum cycle bases are frequently avoided for two reasons: (1) they are computationally expensive and (2) they are not unique. In our MD data, there are relatively few overlapping loops. This means that it is computationally feasible to compute a minimum cycle basis at every time step. As for (2), the lengths of loops in a minimum cycle basis (as a multiset) are unique. We discuss the theory for minimum cycle bases in further detail in Appendix~\ref{app:MCB}.

The complete list of equations relating $\lambda$ and $\{\phi_k\}$ to the appropriate equilibrium constants is the following:
	\begin{align}
		\label{eq:lambdaphi3_start}
		\lambda\phi_3 &= K_3^\text{eff}\bp_3,\\
		\lambda\phi_4 &= K_4^\text{eff}\bp_3^2,\\
		\lambda\phi_5 &= K_5^\text{eff}\bp_3^2,\\
		\lambda\phi_{\ell} &= K_\ell^\text{eff}\lambda\phi_{\ell-1}\frac{\Nc3\pch^2(2-\bp_3)}{\Nh\phh},\quad \ell=6,7,8,\\
		\lambda\phi_{\ell} &= K_L^\text{eff}\lambda\phi_{\ell-1}\frac{\Nc3\pch^2(2-\bp_3)}{\Nh\phh},\quad\ell\ge9.\label{eq:lambdaphi_largeL_end}
	\end{align}

\noindent The equilibrium constants $K^\text{eff}_\ell$, $\ell=3,4,...,8$ and $K^\text{eff}_L$ are given in Table~\ref{tab:Arrh_Param}. 
The loop rate per carbon atom is obtained by summing $\lambda=\sum\lambda\phi_k$, and the loop length distribution is obtained by normalizing $\{\lambda\phi_k\}$. 

The derivation of Eqs.~\eqref{eq:phh_param_start}--\eqref{eq:lambdaphi_largeL_end} is rather complicated. To give the reader a taste of it, we elaborate on the derivation of the degree distribution parameters in the remainder of Sec.~\ref{Sec:Parameter_Fit}. Further details are placed in Appendix~\ref{app:Params}.

\subsection{Degree distribution of the carbon skeleton}

\begin{figure}[!tbp]
	\centering
	\includegraphics{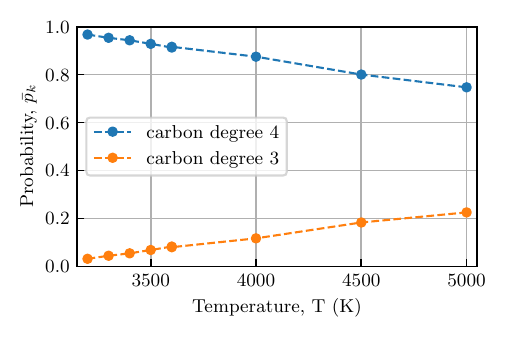}
	\caption{Time-averaged fraction of carbon atoms $\bp_3$ and $\bp_4$ bonded to 3 and 4 atoms, respectively, in ReaxFF MD simulations initialized to \ch{C4H10}. The fraction of carbon atoms bonded to 3 atoms increases with temperature. For computation, we normalize the values $\bp_3$ and $\bp_4$ to sum to 1. In this figure, we do not normalize $\bp_3$ and $\bp_4$ to emphasize that the majority of carbon atoms are bonded to 3 or 4 atoms. This is most accurate at low temperatures.}
	\label{fig:p3_rate}
\end{figure}

\label{sub:skel_deg}
Here we introduce a parametric model for the degree distribution of the carbon skeleton. To do so, we consider the more general problem of modeling the degree distribution of carbon atoms in the global hydrocarbon network split by atom type. We start with the following restriction.
\begin{assume}
	Carbon atoms are bonded to 3 or 4 atoms. Hydrogen atoms are bonded to a single atom.
	\label{asp:HCN_Deg}
\end{assume}
\noindent Carbon atoms that are bonded to 3 atoms are a result of double bonds or radicals. In our MD simulations, Assumption~\ref{asp:HCN_Deg} holds for 98\% of carbon atoms and 99\% of hydrogen atoms. Let $\bp_3$ and $\bp_4$ be the probability a carbon atom is bonded to $3$ or $4$ atoms, respectively. Here $\bp_3 + \bp_4 = 1$. See Fig.~\ref{fig:p3_rate}. The bar is used to distinguish between the global hydrocarbon network and the carbon skeleton, e.g. $p_3$ is the probability a carbon atom is bonded to three other carbon atoms.


Define $\Ncc$, $\Nch$, and $\Nhh$ to be the number of carbon-carbon, carbon-hydrogen, and hydrogen-hydrogen bonds, respectively. To account for the bias of simulated MD networks to have many carbon-hydrogen bonds, let 
\begin{equation}
	\pch=\frac{\Nch}{2\Ncc+\Nch}
\end{equation} 
be the fraction of bonds starting from a carbon atom that lead to a hydrogen atom. We define $\pcc$, $\phc$, and $\phh$ in a similar manner. These probabilities satisfy the following three identities
\begin{equation}
	\pcc+\pch = 1,\quad \phc+\phh=1,\quad (3\bp_3+4\bp_4)\Nc \pch = \Nh \phc.
	\label{eq:pHH_equivs}
\end{equation}
The third identity equates two methods for computing $\Nch$, where $3\bp_3+4\bp_4$ is the mean degree of a carbon atom. By Eq.~\eqref{eq:pHH_equivs}, we only need to compute one bond bias probability. In practice, we use the probability $\phh$, the fraction of H atoms in molecular hydrogen (\ch{H2}).

\begin{algorithm}[t]
	\KwData{$\Nc$ = Number of C nodes,\\ 
		\phantom{\textbf{Data:} }$\Nh$ = Number of H nodes,\\
		\phantom{\textbf{Data:} }$\bar{p}_3$ = fraction of C atoms of degree $3$,\\
		\phantom{\textbf{Data:} }$\phh$ = probability H node is bonded to\\
		\phantom{\textbf{Data:}}\quad another H node}
	\KwResult{Hydrocarbon graph $G$ with unassigned CC bonds}
	Assign $\text{\textit{round}}(\bar{p}_3 \Nc)$ C nodes $3$ stubs\;
	Assign the remaining C nodes $4$ stubs\;
	Assign each H node a single stub\;
	Pair enough H stubs to create $\text{\textit{round}}(\phh \Nh/2)$ edges\;
	Pair the remaining H stubs to C stubs uniformly at random\;
	For each C node extract its degree equal to the number of its stubs that are unassigned\;
	\caption{Estimated Degree Distribution}
	\label{alg:DegDist}
\end{algorithm}

The parameters $\bp_3$ and $\phh$ lead to a natural method for sampling the degree distribution of the carbon skeleton, given in Algorithm~\ref{alg:DegDist}. The parameter $\bp_3$ indicates how many stubs should be assigned to each carbon atom. The parameter $\phh$ assigns the number of hydrogen-hydrogen bonds or \ch{H2} molecules. For carbon-hydrogen bonds, we make the following approximation: all hydrogen atoms that are not bonded to hydrogen are attached randomly to carbon stubs. This approximation accounts for the mean degree of the carbon skeleton, $\langle k \rangle$, through the number of carbon-carbon bonds $\Ncc = \Nc\langle k \rangle/2$. This is a simplification of previous work~\cite{dufour-decieuxPredicting2023} as we do not fit the full degree distribution of the carbon skeleton. Nonetheless, this approximation is accurate because the carbon skeleton has a simple, unimodal degree distribution. See Fig.~\ref{fig:degDist_Ex}. Furthermore, we improve upon previous work~\cite{dufour-decieuxPredicting2023} by considering degree 3 carbon atoms. We remark that Algorithm~\ref{alg:DegDist} does not assign carbon-carbon bonds. Instead, the unassigned carbon stubs are used to sample the degree sequence of the carbon skeleton. This sample degree sequence can be used as input for random graph models, e.g. the configuration model.

\begin{figure}[!tbp]
	\centering
	\includegraphics{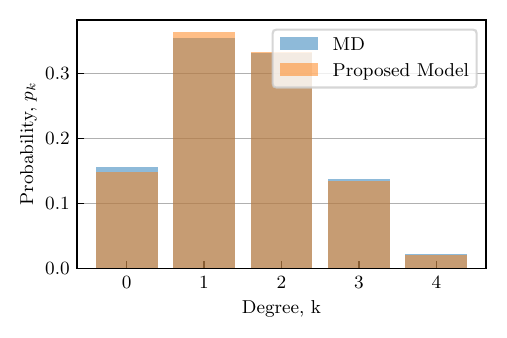}
	\caption{Degree distribution of the carbon skeleton for simulations initialized to \ch{C4H10} at 3600K with $320$ C atoms. The time-averaged degree distribution from MD simulations (blue) has a simple, unimodal shape. This distribution is well-approximated by the binomial-like distribution obtained using the parameters $\phh$ and $\bp_3$ (orange).}
	\label{fig:degDist_Ex}
\end{figure}

\def\pkckh{\bar{p}_{k_{\rm C},k_{\rm H}}}
Under this sampling scheme, we may obtain an explicit formula for the degree distribution of the carbon skeleton. Let $\pkckh$ be the probability a carbon atom is bonded to $k_{\rm C}$ carbon atoms and $k_{\rm H}$ hydrogen atoms. In Algorithm~\ref{alg:DegDist}, each stub of a carbon atom is attached independently to a carbon (hydrogen) atom with probability $\pcc$ ($\pch$). Thus, $\bp_{k_{\ch{C}},k_{\ch{H}}}$ is generated by
\begin{equation}
	\bG(x,y) = \sum_{k_{\ch{C}},k_{\ch{H}}}\bar{p}_{k_{\ch{C}},k_{\ch{H}}} x^{k_{\ch{C}}} y^{k_{\ch{H}}} = \bar{p}_3(\pcc x+\pch y)^3 + \bar{p}_4(\pcc x+\pch y)^4.
	\label{eq:DegDistGlobal}
\end{equation}
The degree distribution of the carbon skeleton is a marginal distribution of $\{\bp_{k_{\ch{C}},k_{\ch{H}}}\}$ generated by
\begin{equation}
	G(x) = \bar{G}(x,1) = p_0 + p_1 x + p_2 x^2 + p_3 x^3 + p_4 x^4
\end{equation}
where
\begin{equation}
	\begin{split}
		p_k &= \bp_3\binom{3}{k} \pcc^k\pch^{3-k}+\bp_4\binom{4}{k} \pcc^k\pch^{4-k},\quad 0{\le}k{\le}3\\
		p_4 &= \bp_4 \pcc^4.\label{eq:DegDist_fromParam}
	\end{split}
\end{equation}
We will primarily use the generating function ${G}(x)$ in our analysis. We note that Algorithm~\ref{alg:DegDist} accurately predicts the degree distribution of the carbon skeleton but not the global hydrocarbon network -- see Appendix~\ref{app:deg_error}, particularly Fig. \ref{fig:Deg_Validation} (right), for further details.   Nonetheless, we suggest the notation $\bar{G}(x,y)$ is useful for studying the hydrocarbon network as a whole. The accuracy of Algorithm~\ref{alg:DegDist} is discussed in further detail in Appendix~\ref{app:deg_error}.



\subsection{Estimation of degree distribution parameters}
\label{sub:deg_param_fit}

Here we fit the degree distribution parameters $\phh$ and $\bp_3$ to MD data. The general trends of these parameters are given in Fig.~\ref{fig:deg_param_fit}. Both $\phh$ and $\bp_3$ increase with temperature, where hydrogen atoms separate from carbon atoms resulting in radicals and double bonds. This does not result in a major change in the average degree of the carbon skeleton. However, when the H/C ratio increases, $\phh$ increases and $\bp_3$ decreases, where the additional hydrogen atoms either form \ch{H2} or bond to carbon. This results in fewer carbon-carbon bonds, so the mean degree of the carbon skeleton decreases with the H/C ratio. This explains why the presence of a giant molecule is primarily a function of the H/C ratio.

\begin{figure}[!tbp]
\centering
\includegraphics{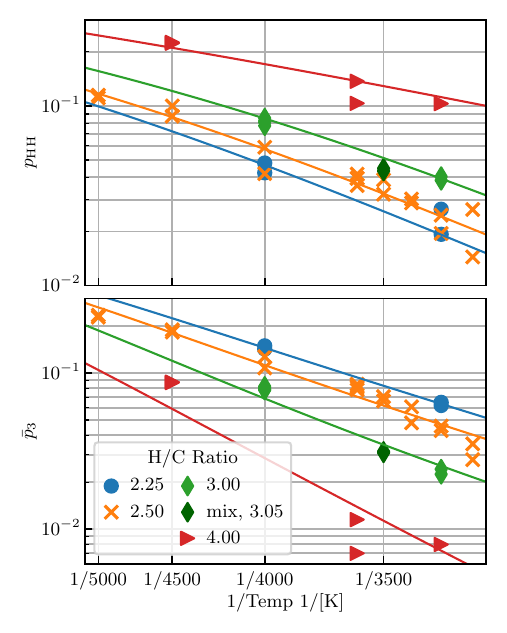}
\caption{
	Parameters $\phh$ (top) and $\bp_3$ (bottom) as estimated by time-averaging MD data (markers) compared to their fit values (solid lines). Specifically, the solid lines are obtained by fitting the equilibrium constants $\KHHE$ and $\KDoubleE$ to an Arrhenius law, which are used to estimate $\phh$ and $\bp_3$ as a function of the H/C ratio and temperature. The Arrhenius laws for $\KHHE$ and $\KDoubleE$ are trained using MD samples initialized with \ch{C4H10}, and they extrapolate to other initial conditions by changing the H/C ratio (via $\Nc$ and $\Nh$) in Eqs.~\eqref{eq:Phh_RR} and~\eqref{eq:P3_RR}. The values $\phh$ and $\bp_3$ can then be used to estimate the degree distribution.
}
\label{fig:deg_param_fit}
\end{figure}

\subsubsection{Estimation of \texorpdfstring{$p_{\rm HH}$}{phh}}

To fit the parameter $\phh$ to MD data, we consider a reaction consisting of two bonds swapping atoms
\begin{equation}
\ch{2 C-H}\quad\ch{<=>}\quad\ch{C-C}\enspace+\enspace\ch{H-H}.
\label{reac:pHH}
\end{equation}
The number of reactants in Reaction~\eqref{reac:pHH} matches the number of products. Thus, we may cancel volumes and write the equilbrium constant in terms of bond counts
\begin{equation}
K_{\ch{HH}} = \frac{[\chemfig{C-C}][\chemfig{H-H}]}{[\chemfig{C-H}]^2} = \frac{\Ncc\Nhh}{\Nch^2}.
\end{equation}

We will express $K_{\ch{HH}}$ as a function of the H/C ratio and $p_{\ch{HH}}$. By Assumption~\ref{asp:HCN_Deg}, each hydrogen participates in a bond. Hence we may write $\Nhh=\Nh \phh/2$ and $\Nch=\Nh\phc = \Nh(1-\phh)$. To count carbon-carbon bonds, we make the approximation $4\Nc \approx 2\Ncc+\Nch$. 
This approximation is not needed for numerical computations, but it allows us to estimate $\phh$ separately from $\bp_3$. Furthermore, this approximation is accurate at low temperatures, where $\bp_3$ is small and the mean degree of carbon atoms is $4-\bp_3\approx4$. Under this approximation, the number of carbon-carbon bonds is $\Ncc \approx (4\Nc - \Nh(1-\phh))/2$ resulting in the effective equilibrium constant
\begin{equation}
\KHHE = \frac{\left(4\frac{\Nc}{\Nh} - (1-\phh)\right) \phh}{(1-\phh)^2} = A_{\ch{HH}} e^{-C_{\ch{HH}}/T}.
\label{eq:Phh_RR}
\end{equation}
Given $\KHHE$, we may rearrange Eq.~\eqref{eq:Phh_RR} into a quadratic equation in $\phh$. This gives two candidate solutions for $\phh$, only one of which is contained in $[0,1]$. This is given by
\begin{equation}
\phh = 1 - \frac{4\frac{\Nc}{\Nh}+1 - \sqrt{\left(4\frac{\Nc}{\Nh}-1\right)^2+16\KHHE\frac{\Nc}{\Nh}}}{2(1-\KHHE)}.
\label{eq:Phh_Sol}
\end{equation}

We measure the percent error of $\phh$ as
\begin{equation}
100\times \Big|\phh-\phh^{(MD)}\Big|/\phh
\end{equation}
where $\phh$ is given by Eq.~\eqref{eq:Phh_Sol} and $\phh^{(MD)}$ is time-averaged from MD simulations. The average percent error is 9.1\% for the \ch{C4H10} training data and 10.7\% for the test data. The total average percent error is 9.9\%. The values $\phh$ and $\phh^{(MD)}$ are depicted in Fig.~\ref{fig:deg_param_fit} (top). We find the amount of molecular hydrogen increases with temperature and H/C ratio.

\subsubsection{Estimation of \texorpdfstring{$\bp_3$}{p3}}

Next, we discuss a method for fitting the parameter $\bp_3$, the probability that a randomly picked carbon is bonded to a total of three other atoms, to MD data. We find that a bond starting from a degree 4 carbon atom is more likely to lead to hydrogen than a bond starting from a degree 3 carbon atom--see Fig.~\ref{fig:Deg_Validation} (right) in Appendix~\ref{app:deg_error} for more details. Let's consider the dissociation of a carbon-hydrogen bond
\begin{equation}
\ch{C4-H} \quad\ch{<=>}\quad\ch{C3}\enspace+\enspace \ch{H}
\label{Reac:P3}
\end{equation}
where $\ch{C}_i$ is a carbon atom bonded to $i$ atoms. The numbers of carbon atoms of degree 3 and 4 are $\Nc\bp_3$ and $\Nc\bp_4$, respectively. By Assumption~\ref{asp:HCN_Deg}, we assume radical hydrogen atoms immediately bond to another atom. Thus, we count the lone hydrogen atoms in Eq.~\eqref{Reac:P3} as the $\Nh \phh$ hydrogen atoms not bonded to carbon. This results in the equilibrium constant
\begin{equation}
K_{\doubleC} = \frac{\Nc \bp_3 \Nh \phh}{\Nc \bp_4 V}
\end{equation}
where $V$ is volume.

We do not have volume data so we assume that the volume is approximately linear in molar mass and temperature under our given conditions, i.e.
\begin{equation}
V \approx c_P (12.011\Nc + \Nh) T.
\label{eq:Vol_Approx}
\end{equation}
Our simulation conditions match that of a supercritical fluid phase. The ability of a supercritical fluid to compress depends on both pressure and temperature. Thus, Eq.~\eqref{eq:Vol_Approx} assumes the compressibility factor $Z:=pV/nRT$ is roughly constant in the conditions studied here. Using this approximation, we define the effective equilibrium constant
\begin{equation}
\KDoubleE = \frac{\bp_3 \Nh\phh}{\bp_4 (12.011 \Nc + \Nh) T} = A_{\doubleC} e^{-{C_{\doubleC}}/{T}}
.
\label{eq:P3_RR}
\end{equation}
Given the Arrhenius coefficients $A_{\doubleC}$ and $C_{\doubleC}$, we may compute $\bp_3$ as follows
\begin{equation}
\bp_3 = \frac{\KDoubleE}{\KDoubleE+\tfrac{\phh}{(12.011\Nc+\Nh)T}}.
\label{eq:P3_Sol}
\end{equation}    
The average percent error of $\bp_3$ is 11.6\% for training data and 16\% for test data. The average percent error over all data points is 13.7\%. The value $\bp_3$ as estimated by Eq.~\eqref{eq:P3_Sol} and by time-averaging MD data are shown In Fig.~\ref{fig:deg_param_fit} (bottom).

\section{Results}
\label{Sec:Case_Study}

\subsection{Test settings}
\subsubsection{Auxiliary models}
Here we measure the ability of the Disjoint Loop Model with Assortativity Correction to predict component sizes. Formally, let \emph{Proposed Model} be the Disjoint Loop Model with Assortativity Correction (as described in Fig.~\ref{fig:model-flowchart}) where the parameters are learned from a local equilibrium model (as described in Sec.~\ref{Sec:Parameter_Fit}). We also introduce three auxiliary models, \emph{Models 1, 2, and 3}. These models {help} to measure the importance of each component of the Proposed Model. These models are ordered by increasing complexity. Model 1 is the configuration model. Model 2 is the Disjoint Loop Model without assortativity correction. Model 3 accounts for assortativity induced by loops using Assortativity Correction. The parameters for Models 1--3 are time-averaged from MD data. For Models 1--3 and Proposed Model, the degree distribution is estimated 
as discussed in Sec.~\ref{sub:skel_deg}. We also compare to previous work~\cite{dufour-decieuxPredicting2023} which uses the configuration model and does not include carbon atoms bonded to 3 atoms. These models are summarized in Table~\ref{tab:RG_models}.

\begin{table*}[t]
	\setlength{\tabcolsep}{1.72pt}
\def\fillertext{\phantom{Deg 3 C atoms}}
\newcommand\tab[2]{\begin{tabular}{c}#1\\#2\end{tabular}}
\centering
\newcommand{\PreserveBackslash}[1]{\let\temp=\\#1\let\\=\temp}
\newcolumntype{C}[1]{>{\PreserveBackslash\centering}p{#1}}
\footnotesize
\begin{tabular}{|C{2.75cm}|C{4cm}|C{4cm}|C{4cm}|}
	\hline
	Model & \tab{Parameter Estimation}{Method} & \tab{Degree Distribution}{Model} & \tab{Random Graph}{Model}
	\\\hline\hline
	Previous Work~\cite{dufour-decieuxPredicting2023}
	& from Arrhenius fit & \tab{a ten-reaction model}{(no deg. 3 C atoms)} & configuration model
	\\\hline
	Model 1 & time-averaged from MD & \tab{using parameters $\phh$}{and $\bp_3$ (Sec.~\ref{sub:skel_deg})} & configuration model
	\\\hline
	Model 2 & time-averaged from MD & \tab{using parameters $\phh$ }{and $\bp_3$ (Sec.~\ref{sub:skel_deg})} & Disjoint Loop Model
	\\\hline
	Model 3 & time-averaged from MD & \tab{using parameters $\phh$}{and $\bp_3$ (Sec.~\ref{sub:skel_deg})} & \tab{Disjoint Loop Model}{+Assortativity Correction}
	\\\hline
	Proposed Model & from Arrhenius fit & \tab{using parameters $\phh$}{and $\bp_3$ (Sec.~\ref{sub:skel_deg})} & \tab{Disjoint Loop Model}{+Assortativity Correction}
	\\\hline
\end{tabular}
\caption{List of random graph models compared in this study. The model from~\cite{dufour-decieuxPredicting2023} is used as a baseline.}
\label{tab:RG_models}
\end{table*}

All of the models listed in Table~\ref{tab:RG_models} can be analyzed using two methods. When the number of nodes is small, we can use random graph sampling (RGS). When sampling random graphs, we make the number of nodes $\Nc,\Nh$ match MD data. In the limit $\Nc\to\infty$, the random graph model can be analyzed via generating functions (GF). For all models, we may evaluate the generating function $H(x)$ with coefficients $P_s$, where $P_s$ is the probability a random node belongs to a connected component of size $s$. {Then the fraction of nodes in the giant component, $S$, and the small component size distribution, $\{\pi_s\}$, are found using Eqs. \eqref{eq:SfromH} and \eqref{eq:PtoPi} respectively.}  

\subsubsection{Loop rate distribution}

\begin{figure}[!tbp]
\centering
\includegraphics{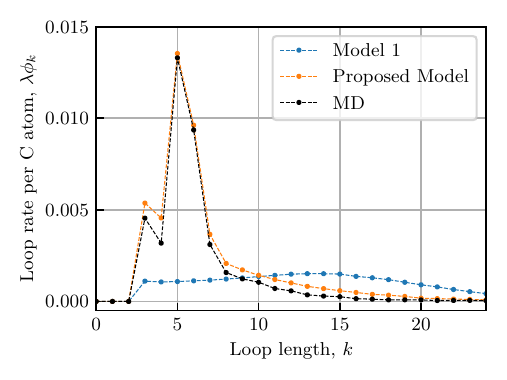}
\caption{
	Loop rates split by length $\{\lambda\phi_k\}$ from MD data compared to Model 1 and Proposed Model for the sample \ch{C8H18} (H/C ratio 2.25). Data for Model 1 and Proposed Model are averaged over 10,000 samples, where $2\Nc^2 \lambda+500$ rewiring steps are used per sample of Proposed Model. Differences between Proposed Model and MD are due to unintended loops and fitting from the local reaction model. 
}
\label{fig:LoopSampling}
\end{figure}

Fig.~\ref{fig:LoopSampling} shows that the Disjoint Loop Model with Assortativity Correction (Model 3 and Proposed Model) accurately recreates the loop rate distribution $\{\lambda\phi_k\}$. In contrast, the configuration model (Model 1) has a very different loop rate distribution which caused poor prediction for the size of the giant component in~\cite{dufour-decieuxPredicting2023}. We also find that Proposed Model has fewer unintended loops than the configuration model. This can be verified by noting that Proposed Model has fewer large loops than Model 1 in Fig.~\ref{fig:LoopSampling}.

\subsubsection{Assortativity Correction setup}

For sampling the Disjoint Loop Model with Assortativity Correction (Model 3 and Proposed Model) one must specify how many rewiring steps are performed by Algorithm~\ref{alg:rewire}. 
The number of steps should increase with the number of loops. 
First, we perform $2\Nc^2\lambda + 500$ rewiring steps and check if the assortativity coefficient is small enough, i.e. $|r|\le 0.01$, averaged over 2000 samples. If not, we perform additional steps until this stopping criterion is satisfied. 
For comparison, the assortativity coefficient of the Disjoint Loop Model can approach values closer to $r=0.1$. See Fig~\ref{fig:AssortComp}. This rewiring scheme can be quite slow. To account for this, we sample the giant connected component multiple times after each initial rewiring scheme. This consists of 50 samples each separated by 50 iterations of edge rewiring. Each of the 50 samples come from roughly the same distribution but they are correlated.

\begin{figure}
\centering
\includegraphics{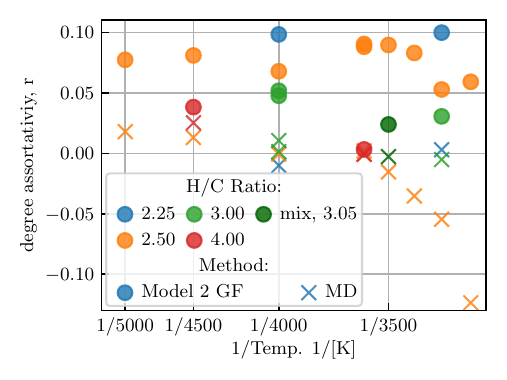}
\caption{Degree assortativity coefficient $r$ of the carbon skeleton. The Disjoint Loop Model (circles) has a higher degree assortativity coefficient than MD simulations (crosses). Thus, we expect a model with neutral assortativity $r=0$ (as in Model 3 and Proposed Model) to be more accurate than a model with positive assortativity $r>0$ (as in Model 2). The assortativity coefficient for Model 2 is obtained using generating functions using Eq.~\eqref{eq:assort_coeff}.}
\label{fig:AssortComp}
\end{figure}

\subsection{Size distribution of the largest molecule}

\begin{figure}[!tbp]
\centering
\includegraphics[width=\textwidth]{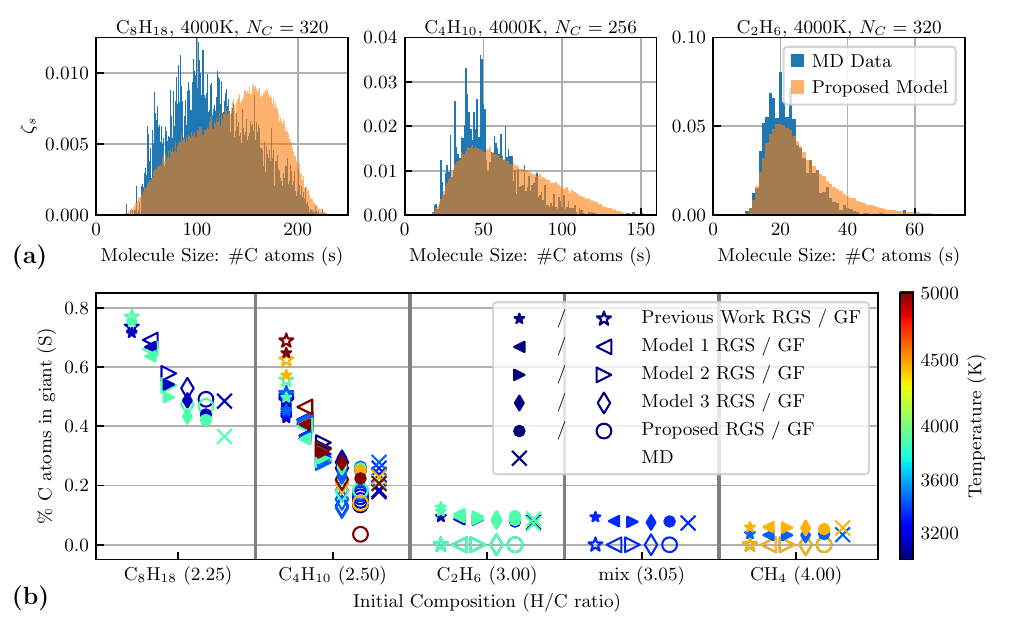}
\caption{(a) Sample histograms for $\zeta_s$ the probability the largest molecule has $s$ carbon atoms for Proposed Model versus MD data. All histograms in (a) are obtained by random graph sampling. The ten-reaction model from Previous Work~\cite{dufour-decieuxPredicting2023}, Model 1, and Model 2 are averaged over 100,000 samples. Model 3 and Proposed Model are averaged over 20,000 initial rewire runs followed by 50 samples each. 
	(b) Expected fraction of carbon atoms in the largest molecule under all random graph models compared to MD data. Samples are organized by increasing H/C ratio. Solid markers correspond to the sampling of the largest molecule and unfilled markers correspond to $1-H(1)$. When the H/C ratio is 3 or higher all random graph models predict $1-H(1)=0$ or that there is no giant connected component in the limit $\Nc\to\infty$. }
\label{fig:GiantComps}
\end{figure}

First, we estimate the largest molecule size distribution. 
All of the random graph models listed in Table~\ref{tab:RG_models} suggest there is a phase transition of hydrocarbon pyrolysis between an H/C ratio of 2.5 and 3 for the temperatures and pressure given here. When the H/C ratio is 3 or higher, no giant component is predicted ($H(1)=1$). 
For at an H/C ratio between 2.5 and lower, a giant component is predicted ($H(1)<1$). 

To measure the largest molecule size, we report the fraction, $S$, of carbon atoms that are in the largest molecule. In Fig.~\ref{fig:GiantComps}(b), we show the expected value of $S$ for all random graph models and MD data. When there is no giant connected component (H/C ratio 3 or higher) all models {accurately} predict the size of the largest molecule. When the giant component is present (H/C ratio 2.5 or lower), the configuration model overestimates the size of the largest molecule. In this regime, Models 1, 2, and 3 incrementally decrease the estimated size of the giant connected component. The size of the largest molecule is estimated reasonably well by Model 3 and Proposed Model even when the H/C ratio is 2.5 or lower. 
In Fig.~\ref{fig:GiantComps}(a), we show the distribution of the number of carbon atoms in the largest molecule, $\{\zeta_s\}$, sampled from Proposed Model versus averaged MD data. {Additional plots of the largest molecule size distribution are given in Appendix~\ref{app:extra_figs}.} 
We observe that $\{\zeta_s\}$ as estimated by Proposed Model has a heavier tail than ReaxFF MD data. 
Nevertheless, for most samples, we find that Proposed Model matches $\{\zeta_s\}$ from MD data surprisingly well.
The exception is the case of \ch{C8H18}, e.g. see the left histogram in Fig.~\ref{fig:GiantComps}(a).
The reason explaining the discrepancy between the actual and predicted distributions of the giant component at low H/C ratio such as 2.25 is that there are many loops, and they do overlap. Hence, Assumption~\ref{asp:DisLoop} of loops being disjoint underlying Proposed Model is not valid. We leave the elimination of Assumption~\ref{asp:DisLoop} for future work.

\begin{figure}[!tbp]
\centering
\includegraphics[width= \textwidth]{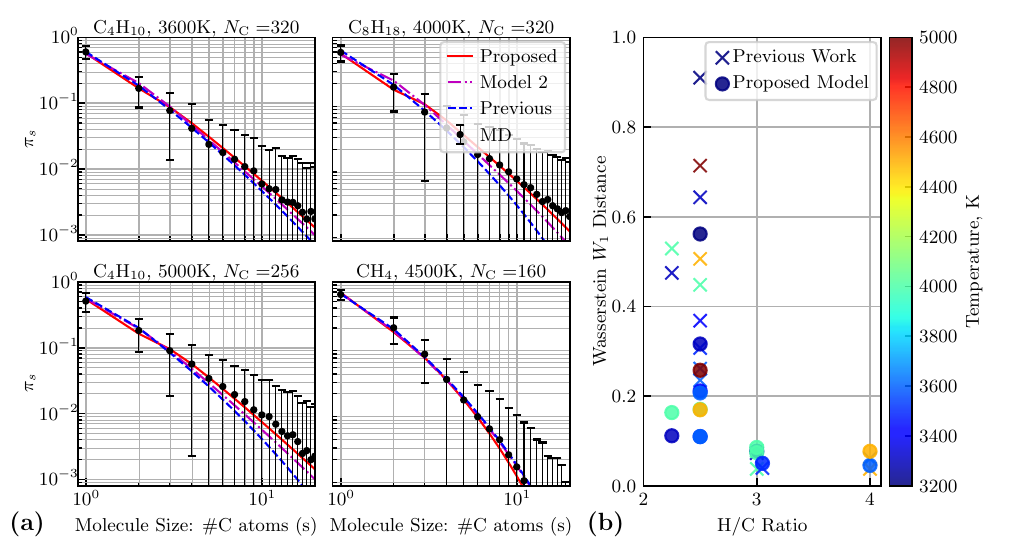}
\caption{Analysis of the distribution of the number of carbon atoms in a random molecule $\{\pi_s\}$. (a) Sample distribution of $\{\pi_s\}$ given by Proposed Model, Model 2, Previous Work~\cite{dufour-decieuxPredicting2023}, and MD data. All samples of $\{\pi_s\}$ for Previous Work~\cite{dufour-decieuxPredicting2023}, Model 2, and Proposed Model are given by generating functions. One standard deviation of $\{\pi_s\}$ is shown for samples from MD data. 
	(b) Wasserstein distance of $\{\pi_s\}$ as estimated by Previous Work~\cite{dufour-decieuxPredicting2023} (crosses) and Proposed Model (circles) versus $\{\pi_s\}$ time-averaged from MD data.}
\label{fig:SmallComps}
\end{figure}

\subsection{Size distribution of small molecules }
Next, we measure how well the random graph models predict the size distribution of small molecules. We choose 20 carbons to be the cutoff size for small molecules. To achieve this, $\pi_s$ is evaluated for $s=1,\hdots,20$ and normalized to sum to 1. {The Wasserstein $W_1$ distance~\footnotemark[1] measuring the discrepancy between the small molecule size distribution of Proposed Model and the MD data is displayed in Fig.~\ref{fig:SmallComps}(b). 
\footnotetext[1]{The Wasserstein $W_1$ distance between probability mass functions $\{P_j\}$ and $\{Q_j\}$ is $W_1(P,Q)=\sum_s\left|\sum_{j=1}^s P_j - \sum_{j=1}^s Q_j\right|$}
For comparison, the similar result obtained in the previous work~\cite{dufour-decieuxPredicting2023} is also shown there.} We find that both models predict $\{\pi_s\}$ reasonably well when the giant component is absent. When the giant component is present, Proposed Model predicts the size distribution of small molecules better than the model in~\cite{dufour-decieuxPredicting2023}. Sample small molecule size distributions are given in Fig.~\ref{fig:SmallComps}(a) and the rest are in Appendix~\ref{app:extra_figs}. Note the impact of Assortativity Correction on the molecule size distribution. 
In assortative networks, more edges connect pairs of low-degree nodes. As a result, Model 2 (the Disjoint Loop Model without Assortativity Correction) overestimates the fraction of molecules with two carbon atoms ($\pi_2$), i.e. the number of bonds connecting carbon atoms with skeleton degree 1. See the purple line in the top-right plot in Fig.~\ref{fig:SmallComps}(a). 



\section{Conclusion}
\label{sec:conclusion}

We proposed a random graph model for estimating molecule size distributions in hydrocarbon pyrolysis, the Disjoint Loop Model with Assortativity Correction. The input for this model is obtained from a series of local equilibrium reaction models trained on MD data. This model has the following advantages.
\begin{enumerate}
\item \textit{Low computational cost.} Using this model, one can predict the size of the largest molecule and the small molecule size distribution in hydrocarbon pyrolysis, at any H/C ratio from 2.25 to 4, and any temperature between 3200K and 5000K at a very low computational cost. Specifically, if the number of atoms is roughly 1,000 then 100,000 random graph samples are generated within half an hour on a laptop with an Intel i7 10th gen. The prediction for the fraction for the largest molecule size and small molecule size distribution in the limit $\Nc\to\infty$ is obtained using generating functions and solving a nonlinear system of equations within a second.

\item \textit{Small number of required input parameters.} Our model uses 9 reaction equilibrium constants learned from MD data. In contrast, the local kinetic Monte Carlo model that predicted the size of the largest molecule comparably~\cite{dufour-decieuxTemperature2022} needed 154 local reaction rates. 

\item \textit{The largest molecule.} The model accurately predicted the molecule size distribution for H/C ratio $\ge 2.5$ where the loops were mostly non-overlapping. It somewhat exaggerated the largest molecule size at H/C ratio at 2.25 as in \ch{C8H18}.

\item \textit{The small molecule size distribution.} The model accurately predicted the small molecule size distribution for the whole range of H/C ratios and temperatures.
\end{enumerate}

The term disjoint in the Disjoint Loop Model with Assortativity Correction refers to the strong assumption (Assumption~\ref{asp:DisLoop}) that the loops are disjoint. This assumption enabled us to use existing methods~\cite{karrerRandom2010} to derive the molecule size distribution and largest component size.
This assumption implies that the small connected components are tree-like at the level of motifs (edges and small loops). 
But to what extent is Assumption~\ref{asp:DisLoop} true?  In hydrocarbon networks initialized with \ch{CH4}, \ch{C2H6}, \ch{C4H10}, and \ch{C8H18}, over 98\%, 90\%, 75\%, and 66\% of biconnected components are composed of a single loop. Thus, when there is no giant connected component (H/C ratio $\ge3$), most biconnected components are accounted for. When there is a giant connected component (H/C ratio $\le 2.5$), other biconnected components appear. These mostly consist of a small number of loops (2 to 4) that are connected in a tree-like fashion. Examining Fig.~\ref{fig:GiantComps} we observe that the size of the largest molecule for the system initialized by C$_8$H$_{18}$ (H/C ratio 2.25) is predicted the least accurately. This means that more carbon-rich systems cannot be adequately described by the Disjoint Loop Model with Assortativity Correction. To model such systems, we plan to generate more ReaxFF MD data and develop a new random graph model that abandons the disjoint loop assumption.  

Finally, we believe that hydrocarbon pyrolysis can be studied using random graph theory because of its extreme conditions. The high temperatures and high pressure of hydrocarbon pyrolysis make it ergodic so many molecular configurations can form. Thus, we theorize that these methods may be extendable to other chemical systems at extreme conditions, such as combustion. Multivariate distributions and multivariate generating functions may be instrumental for their modeling by random graphs. 


\section*{Acknowledgements} 
This material is based upon work supported by the National Science Foundation Graduate Research Fellowship Program under Grant No. DGE 2236417 (PR). Any opinions, findings, and conclusions or recommendations expressed in this material are those of the author(s) and do not necessarily reflect the views of the National Science Foundation.
This work was partially supported by AFOSR MURI grant FA9550-20-1-0397 (MC). V.\ D.-D.\ acknowledges the Rütli Stiftung and the ETH Zurich Foundation.

\bibliographystyle{ieeetr}
\bibliography{bib_Zotero.bib}

\newpage

\appendix

\section{Loops in the Configuration Model}
\label{app:Conf_Loop}
Here we compute the expected number of loops of fixed length $L\ge 3$ that are present in the configuration model.  The calculations carried out here use methods from \cite[Chapter 12]{newmanNetworks2018}. We assume that the degree distribution is fixed with finite second moment, $\langle k^2\rangle<\infty$. The expected number of loops of length $L$ approaches a constant as $N\to\infty$. Thus, the rate of loops of length $L$ per node decays as $O(N^{-1})$. In Appendix~\ref{app:Triangles_DLM} we show how these computations generalize to the Disjoint Loop Model

We start by computing edge probabilities. In the configuration model, each pair of stubs is equally likely to form an edge. The total number of stubs is $N\langle k\rangle$ where $\langle k\rangle$ is the mean degree. Then, the probability two specific stubs are attached by an edge is
\begin{equation*}
\mathbb{P}({\rm stub}_i\sim{\rm stub}_j) = \frac{1}{N\langle k\rangle-1} = \frac{1}{N\langle k\rangle}+O(N^{-2}).
\end{equation*}
Let $u$ and $v$ be two nodes with fixed degrees $k_u$ and $k_v$. Define $u\sim v$ to be the event that $u$ and $v$ are joined by an edge. There are $k_u k_v$ ways we may select a stub from $u$ and $v$ and the probability that $u$ and $v$ are joined by two or more edges is $O(N^{-2})$. Thus, $u$ and $v$ are joined by an edge with probability
\begin{equation}
\mathbb{P}(u\sim v) = \frac{k_u k_v}{N\langle k\rangle}+O(N^{-2}).
\end{equation}
Let $w$ be a third node with fixed degree $k_w$. Conditioned on the edge $u\sim v$, node $v$ only has $k_v-1$ stubs that may attach to node $w$. This happens with probability
\begin{equation}
\mathbb{P}(v\sim w | u\sim v) = \frac{(k_v-1)k_w}{N\langle k\rangle} + O(N^{-2}).
\end{equation}

Now, we may compute the probability that $u$, $v$, and $w$ form a triangle. Conditioned on the edges $u\sim v\sim w$, the nodes $u$ and $w$ have $k_u-1$ and $k_w-1$ stubs that may be attached, respectively. Thus, the probability of a triangle is
\newcommand\hone{\hspace{0pt}}
\newcommand\hpar{\hspace{0pt}}
\begin{equation}
	\begin{split}
	\mathbb{P}(\hone u \sim v\sim w\sim u\hone) &= \mathbb{P}\hone(\hpar u\sim v\hpar)\hone \mathbb{P}\hone(\hpar v\sim w\hone|\hone u\sim  v\hpar)\hone \mathbb{P}\hone(\hpar w\sim u\hone|\hone u\sim v\sim w\hpar)\\&= \frac{k_u k_v}{N\langle k\rangle}\frac{(k_v- 1)k_w}{N\langle k\rangle}\frac{(k_w - 1)(k_u- 1)}{N\langle k\rangle} + O(N^{-4}).
	\end{split}
	\label{eq:tri_spec_B1}
\end{equation}
If the nodes $u,v,w$ are sampled randomly, then the degrees $k_u,k_v,k_w$ are independent samples from the degree distribution $\{p_k\}$. Thus, summing Eq.~\eqref{eq:tri_spec_B1} over node triples gives the expected number of triangles
\begin{equation}
	\mathbb{E}[\#\bigtriangleup] = \binom{N}{3}\left(\sum_k p_k \frac{k(k-1)}{N\langle k\rangle}\right)^3 + O(N^{-1})
	=\frac{1}{6}\left(\frac{\langle k^2\rangle - \langle k\rangle}{\langle k\rangle}\right)^3 + O(N^{-1}).
	\label{eq:exp_Tr_1_B1}
\end{equation}

The probability $L$ nodes form a loop $u_1\sim u_2\sim \hdots \sim u_L\sim u_1$ can be obtained by extending the product in Eq.~\eqref{eq:tri_spec_B1} to obtain the leading term:
\begin{equation}
\prod_{i=1}^L \frac{k_{u_i}(k_{u_i}-1)}{N\langle k \rangle}.
\end{equation}
The expected number of loops of length $L$ is
\begin{align}
\mathbb{E}[\#\text{$L$-loops}] & = \binom{N}{L}\frac{L!}{2L}\left(\sum_k p_k \frac{k(k-1)}{N\langle k\rangle}\right)^L + O(N^{-1})\notag\\
&=\frac{1}{2L}\left(\frac{\langle k^2\rangle - \langle k\rangle}{\langle k\rangle}\right)^L + O(N^{-1})&.
\label{eq:appA_exp_loop_L}
\end{align}
where $L!$ is the number of orderings of $L$ nodes and $2L$ is the number of automorphisms of a loop of length $L$.

This result can also be obtained using expected excess degrees. Let $u$ be a node reached by a random stub. The probability $u$ has $k$ excess neighbors is $q_k=(k+1)p_{k+1}/\sum_{k'} k'p_{k'}$ with mean
\begin{equation}
\tau = \sum_k kq_k = \frac{\langle k^2\rangle - \langle k\rangle}{\langle k\rangle}.
\end{equation}
The excess neighbors of $u$ also reached by a random edge, so their excess degree has distribution $\{q_k\}$. The expected number of second neighbors to $u$ reached by following excess edges is $\tau^2$. In general, the expected number of nodes distance $L$ to $u$ reached by following excess edges is given by $\tau^L$. These $L$'th excess neighbors are also reached by a random edge. The probability a particular one of these edges leads to the stub used to reach $u$ (resulting in a loop) is approximately $1/N\langle k\rangle$. There are $N\langle k\rangle$ stubs that could have been used to draw $u$. These quantities cancel so---accounting for automorphisms---the expected number of loops of length $L$ is approximately $\tau^L/2L$ matching Eq.~\eqref{eq:appA_exp_loop_L}.

\section{Probabilistic Ingredients for the Disjoint Loop Model}

\subsection{Proof of Lemma 1}
\label{app:Inc_Prob}
{ 
The proof of Lemma~\ref{lemma:inclusion_prob} requires the verification of the asymptotic loop inclusion probabilities given by Eq.~\eqref{eq:WSWRIncProb}. It is challenging to prove Eq.~\eqref{eq:WSWRIncProb} using the classical scheme for weighted sampling without replacement. Instead, we prove Eq.~\eqref{eq:WSWRIncProb} using an equivalent sampling scheme that utilizes order statistics introduced in Ref.~\cite{efraimidisWeighted2006}:
\begin{enumerate}
	\item for each node $i$ draw a random number (key) in $[0,1]${  $$y_i=\begin{cases}U_i^{1/w_{k_i}}, & w_{k_i}>0\\ 0, &w_{k_i}=0\end{cases},$$ where $k_i$ is the degree of the node $i$, $w_{k_i} = \binom{k_i}{2}$, and $U_i\sim\text{Uniform}[0,1]$}
	\item return the {$\mathcal{N} := \sum_{j=1}^{N_L}L_j$} nodes with the largest keys $y_i$.
\end{enumerate}
The node degrees $\{k_i\}$ and thus the weights $\{w_{k_i}\}$ are assumed to be i.i.d., so the fraction of nodes with keys $y_i > a$ approaches a fixed value as $N\to\infty$. We will show a node of degree $k\ge 2$ is in a loop if its key is sufficiently large as $N\to\infty$, which happens with probability
\begin{equation}
	\mathbb{P}(U^{1/w_k} > q) = \mathbb{P}(U>q^{w_k})=1-q^{w_k}
\end{equation}
matching Eq.~\eqref{eq:WSWRIncProb}.

We remark, the time complexity of weighted sampling without replacement is small compared to random graph sampling, e.g. the Assortativity Correction step. Nevertheless, the above sampling scheme is readily implemented by sorting the keys $\{y_i\}$ to find the $\mathcal{N}$ maximizers in $O(N \ln(N))$ time or by using the optimizations discussed in Ref.~\cite{efraimidisWeighted2006}.

The following proof of Lemma~\ref{lemma:inclusion_prob} is adapted from~\cite[Lemma 2.1]{rosenInclusion2000}, with minor modifications listed here. First, the number of nodes in loop motifs, $\mathcal{N}$, and weights $\{w_{k_i}\}$ are random in our setting while they are fixed in~\cite{rosenInclusion2000}. Second, \cite{rosenInclusion2000} implements weighted sampling without replacement using order statistics of exponential random variables, where only small values are accepted. However, the methods in~\cite[Lemma 2.1]{rosenInclusion2000} are suitable for the sampling scheme~\cite{efraimidisWeighted2006} which we will use in the following proof. Last, we do  not discuss convergence rates here. The reader can find this discussion in~\cite{rosenInclusion2000}.

\begin{proof}
	Let Eq.~\eqref{eq:Loop_Feasibility} hold. Then Eq.~\eqref{eq:q_identity} has a unique solution $q\in(0,1)$. Let $i$ be an arbitrary node with random degree $k_i$. Ultimately, we want to compute the limit
	\begin{equation}
		f(k) = \lim_{N\to\infty}\mathbb{P}_k(\text{node $i$ in a loop motif}),
		\label{eq:f_lim}
	\end{equation}
	if the limit exists, where $\mathbb{P}_k(\cdot):=\mathbb{P}(\cdot|k_i=k)$.
	
	Let $y^*$ be the minimum accepted key, i.e. the smallest random number selected 
	at the end of step 2. We will show that $y^*\rightarrow q$ as $N\rightarrow\infty$.
	For finite $N$, node $i$ is in a loop motif if its key is at least as large as $y^*$, corresponding to the event $\{y_i\ge y^*\}$. For any $a\in\RR$,
	\begin{equation}
			\mathbb{P}_k(y_i\ge y^*) \ge \mathbb{P}_k(y_i>a,y^*\le a)
			= \underbrace{\mathbb{P}_k(y_i>a)}_{I} - \underbrace{\mathbb{P}_k(y_i>a,y^*> a)}_{II}.
		\label{eq:B.1_prob_split}
	\end{equation}
	We choose $a$ to be slightly larger than the value $q$, i.e. $a=q+\delta$ for sufficiently small $\delta >0$. For this value of $a$, term $I$ will be close to the desired limit given by Eq.~\eqref{eq:WSWRIncProb},
	\begin{equation}
		I = \mathbb{P}_k(U_i^{1/w_k}>q+\delta)=1-(q+\delta)^{w_k},
		\label{eq:term_I}
	\end{equation}
	where $U_i$ is independent of $k_i$.

	We will now show the term $II$ in Eq.~\eqref{eq:B.1_prob_split} goes to $0$ as $N\to\infty$. It is sufficient to show that $y^*$ will not deviate far from its limit (i.e. $\mathbb{P}(y^*>q+\delta)\to0$ as $N\to\infty$) because
	\begin{equation}
			\mathbb{P}_k(y_i>a,y^*>a) = {\mathbb{P}(y_i>a,y^*>a,k_i=k)}/{p_k}\le \mathbb{P}(y^*>a)/p_k.
	\end{equation}
	The event $\{y^*>a\}$ is equivalent to stating that $\mathcal{N}$ nodes have keys larger than $a$, i.e.
	\begin{equation}
		\{y^*>a\}\equiv\left\{\sum\nolimits_{i=1}^N \bm{1}[y_i>a]\ge \mathcal{N}\right\},
		\label{eq:ys_rewrite}
	\end{equation}
	where $\bm{1}$ is the indicator function. The loop lengths are i.i.d. so the moments of their sums are given by
	\begin{align*}
		\EE[\mathcal{N}] &= \EE N_L\EE L = \lambda N\langle\phi_k\rangle,\\
		\EE\left[ \mathcal{N}^2\right] &= \sum_{n=1}^\infty P(N_L=n)\underbrace{\EE\left[\left(\sum\nolimits_{j=1}^{n}L_j\right)^2\right]}_{\text{Var}({\sum L_j})+\EE[\sum L_j]^2}
		=\lambda N\text{Var}L+(\lambda N+(\lambda N)^2)\langle\phi_k\rangle^2,\\
		\text{Var}\left(\mathcal{N}\right) &= \lambda N(\text{Var}L+\langle\phi_k\rangle^2).
	\end{align*}
	By Chebychev's inequality, the fraction of nodes in loop motifs will converge to its mean $\lambda \langle\phi_k\rangle$ in probability:
	\begin{equation}
		\mathbb{P}\left(\left| \frac{\mathcal{N}}{N}
		-\lambda \langle\phi_k\rangle\right|\ge \epsilon\right)\le \frac{\lambda (\text{Var}L+\langle\phi_k\rangle^2)}{\epsilon^2 N}.
		\label{eq:loop_count_sum}
	\end{equation}
	The keys $y_i$ are i.i.d., so by the Law of Large Numbers,
	\begin{align}
		\lim_{N\rightarrow\infty}\frac{1}{N}\sum_{j=1}^N \bm{1}[y_j>a] &= P(y_j>a) 
		= \sum_k p_k (1-a^{w_k}).
		\label{eq:drawn_nodes_conc}
	\end{align}
	The functions $1-a^{w_k}$ are decreasing on $[0,1]$, so
	\begin{equation}
		\sum_k p_k (1-(q+\delta)^{w_k})<\sum_k p_k (1-q^{w_k})=\lambda\langle\phi_k\rangle,
		\label{eq:key_sum}
	\end{equation}
	where the final equality is due to Eq.~\eqref{eq:q_identity}.  Eq.~\eqref{eq:key_sum} shows that
	\begin{equation}
		\label{eq:B9prime}
		\lim_{N\rightarrow\infty}\frac{1}{N}\sum_{j=1}^N \bm{1}[y_j>q+\delta]<
		\lambda \langle\phi_k\rangle =  \frac{\mathbb{E}[\mathcal{N}]}{N}.
	\end{equation}
	
	\noindent Furthermore, Eqs.~\eqref{eq:loop_count_sum} and~\eqref{eq:drawn_nodes_conc} tell us that the random variables $\mathcal{N}/N$ and $N^{-1}\sum_{j=1}^N \bm{1}[y_j>q+\delta]$ concentrate around their means. This concentration and the fact that the inequalities in Eqs.~\eqref{eq:ys_rewrite} and \eqref{eq:B9prime} are opposite imply that the probability of the event in Eq.~\eqref{eq:ys_rewrite} tends to zero as $N\rightarrow\infty$, i.e.
	\begin{equation}
		\lim_{N\rightarrow\infty}\mathbb{P}\left[\frac{1}{N}\sum\nolimits_{i=1}^N \bm{1}[y_i>q+\delta]\ge \frac{\mathcal{N}}{N}\right] = 0.
	\end{equation}
	
	\noindent Hence, term $II$ in Eq.~\eqref{eq:B.1_prob_split} tends to zero as $N\to\infty$.
	Thus, taking Eq.~\eqref{eq:B.1_prob_split} as $N\to\infty$ with $a=q+\delta$ gives 
	\begin{equation}
		\liminf_{N\to\infty} P_k(y_i\ge y^*)\ge 1-(q+\delta)^{w_k}.
		\label{eq:B.1_liminf}
	\end{equation}
	Taking the supremum of Eq.~\eqref{eq:B.1_liminf} over $\delta>0$ gives
	\begin{equation}
		\liminf_{N\to\infty} P_k(y_i\ge y^*)\ge 1-q^{w_k}.
		\label{eq:B.1_lower}
	\end{equation}
	
	The argument for the upper bound is similar by taking $a=q-\delta$ to bound $\mathbb{P}_k(y_i< y^*) = 1-\mathbb{P}_k(y_i\ge y^*)$ where
	\begin{equation}
		\mathbb{P}_k(y_i < y^*) \ge \mathbb{P}_k(y_i\le a, y^*>a). 
	\end{equation}
	Combined, Eq.~\eqref{eq:B.1_lower} and the upper bound give
	\begin{equation}
		\begin{split}
			1-q^{w_k}&\le\liminf_{N\to\infty} P_k(y_i\ge y^*)\le \limsup_{N\to\infty} P_k(y_i\ge y^*)\le1-q^{w_k}.
		\end{split}
	\end{equation}
	This guarantees that the limit in Eq.~\eqref{eq:f_lim} exists and $f(k) = \lim_{N\to\infty}P_k(y_i\ge y^*)=1-q^{w_k}$ as claimed.
\end{proof}
}

\subsection{Extra Loops are Rare}
\label{app:Rare_Loops}
{
Here we verify the sufficiency condition for extra loops to be rare in the Disjoint Loop Model. In particular, let $\{p_k\}$ and $\{\phi_k\}$ have finite second moments. Then we want to show the number of edge motifs attached to a random node not in a loop motif or a random loop motif has finite second moments. 

We start by considering a random node that is not in loop motif. The probability this node has degree $k$ is $\tp_{k,0}/\sum_{k'}\tp_{k',0}$ where $\tp_{k,0}=(1-f(k))p_k$. This has second moment
\begin{equation}
	\label{eq:k2pk0}
	\frac{\sum_k k^2(1-f(k))p_k}{\sum_{k'} (1-f(k'))p_{k'}}\le\frac{\langle k^2\rangle}{\sum_{k'}(1-f(k'))p_{k'}}
\end{equation}
where $0\le f(k)\le 1$. The right-hand side of Eq. \eqref{eq:k2pk0} is necessarily finite as long as the fraction of nodes not in loop motifs satisfies $\sum_{k'} p_{k',0}>0$.

Next, we consider a random loop motif. The number of edge motifs attached to a node in a loop motif is generated by the excess degree distribution $\tG^L(x)$. The number of edge motifs attached to a loop of length $\ell$ is generated by $[\tG^L(x)]^\ell$. By conditioning on loop length, the number of edge motifs attached to a loop motif is generated by $\Phi_0(\tG^L(x))$. See Table~\ref{tab:Gen_Funcs} for the definition of these generating functions. To calculate the second moment of the PMF generated by $\Phi_0(\tG^L(x))$ we use the following identity relating the second moment of a PMF $\{g_k\}$ to its probability generating function $\mathcal{G}(x)$:
\begin{equation}
	\frac{d}{dx}[x\mathcal{G}'(x)]_{x=1} = \sum k^2 g_k.
\end{equation}
Thus, the second moment of the number of edge motifs attached to a random loop motif is given by
\begin{equation}
	\Phi_0'(1)[\tG^L_x(1)+\tG^L_{xx}(1)] + \Phi_0''(1)[\tG^L_x(1)]^2. 
\end{equation}
These second derivatives can be bounded using the second moments of $\{p_k\}$ and $\{\phi_k\}$
\begin{align}
	\Phi_0''(1) &= \sum_k k(k-1)\phi_k,\\
		\tG^L_{xx}(1) &= \frac{\tG_{xxz}(1)}{\tG_z(1)} = \frac{\sum_k f(k+2)k(k-1)p_{k+2}}{\tG_z(1)}
		\le \frac{\langle k^2\rangle}{\tG_z(1)}.\label{eq:Gxxz}
\end{align}
The identity in Eq.~\eqref{eq:Gxxz} uses
\begin{equation}
	f(k+2)k(k-1)\le k(k-1) \le (k+2)^2.
\end{equation}
Eq.~\eqref{eq:Gxxz} is well-defined if the expected number of nodes in loops $\tG_z(1)$ is non-zero.
}

\subsection{Extra Triangle Counts}
\label{app:Triangles_DLM}
\newcommand\simL{\overset{L}{\sim}}

Here we compute the expected number of triangles in the Disjoint Loop Model that are not explicitly introduced as loop motifs. Ultimately, we express these counts in terms of the generating functions listed in Table~\ref{tab:Gen_Funcs}. We split this into cases by the motifs used to construct the triangle. Let $u\sim v$ and $u\simL v$ be the events that nodes $u$ and $v$ are joined by an edge in an edge motif and loop motif, respectively. Let $f(n)\simeq g(n)$ imply that $f(n)$ is asymptotically equal to $g(n)$, i.e. $f(n)/g(n)\to1$ as $n\to\infty$.

\textit{Case 1.} We first compute the expected number of triangles that are formed by three edge motifs. The total number of stubs for edge motifs is
\begin{equation}
N\sum\nolimits_k k \tp_{k,\ell} = N\tG_x(1,1).
\end{equation}
Each pair of stubs for edge motifs (i.e. after loop nodes are assigned) is equally likely to be connected by an edge. Let $u$, $v$, and $w$ be nodes that participate in $k_u$, $k_v$, and $k_w$ edge motifs, respectively. By the same reasoning used to reach Eq.~\eqref{eq:tri_spec_B1} they form a triangle with probability
\begin{equation}
	P(u\sim v\sim w\sim u) \simeq \frac{k_u k_v}{N\tG_x(1,1)}\frac{(k_v-1)k_w}{N\tG_x(1,1)}\frac{(k_w-1)(k_u-1)}{N\tG_x(1,1)}.
\label{eq:tr_Pr_B2}
\end{equation}
Then, by summing Eq.~\eqref{eq:tr_Pr_B2} over triples gives the expected number of triangles composed of only edge motifs
\begin{equation}
\binom{N}{3}\left(\frac{\sum_k k(k{-}1)\tp_{k,\ell}}{N\tG_x(1,1)}\right)^3 =\frac{[\tG^E_x(1,1)]^3}{6}.\label{eq:DLM_LoopCount1}
\end{equation}
as $N\to\infty$. Note stubs for edge motifs are combined pairwise at random as in the configuration model. In particular, one could derive Eq.~\eqref{eq:DLM_LoopCount1} from Eq.~\eqref{eq:exp_Tr_1_B1} where $\tG^E_x(1,1)=\tG_{xx}(1,1)/\tG_x(1,1)$ is the expected number of excess edge motifs a node reached by a random edge motif has.

\textit{Case 2.} Next, we consider the number of extra triangles that consist of two edges in a loop motif joined by a regular edge.
Let $u\simL w\simL v$ be a connected triple in a loop motif whose endpoints have $k_u$ and $k_v$ stubs for edge motifs. The probability these nodes form a triangle is then
\begin{equation}
P(u\sim v|u\simL w\simL v) \simeq \frac{k_u k_v}{N\tG_x(1,1)}.
\label{eq:appB2_C2_1Mid}
\end{equation}
Nodes are shuffled before they are assigned to loops. Thus, the degrees $k_u$ and $k_v$ of nodes $u$ and $v$ in a random triple $u\simL w\simL v$ are independent (ignoring finite sample space effects which can be ignored in the limit $N\to\infty$). In particular, $k_u$ and $k_v$ have PMF $\{q_k^L\}$, where $q_k^L$ is the probability a node in a loop motif has excess degree $k$, generated by $\tG^L(x)$. There is a unique triple $u\simL w\simL v$ for each node $w$ in a loop motif, and the expected number of nodes in loop motifs is $\lambda N \langle \phi_k\rangle$. Finally, the number of triangles of the form $u\simL w\simL v\sim u$ is obtained from Eq.~\eqref{eq:appB2_C2_1Mid} by taking the expectation of $k_u$ and $k_v$ and multiplying by the expected number of loop nodes. This is given by
\begin{equation}
\lambda N \langle \phi_k\rangle\frac{(\sum_{k}k \tq_k^L)^2}{N\tG_x(1,1)} = \lambda\Phi_0'(1) \frac{[\tG^L_x(1)]^2}{\tG_x(1,1)}.\label{eq:app_Tr_DLM_2.0}
\end{equation}
	By Eq.~\eqref{eq:q_identity_GF}, $\lambda\Phi_0'(1)$ may be replaced by $\tG_z(1)$. This cancels with the denominator of a single factor $\tG^L_x(1)=\tG_{xz}(1)/\tG_z(1)$ in Eq.~\eqref{eq:app_Tr_DLM_2.0}. The expected number of connected triples in loop motifs that are closed into a triangle by an edge motif is
	\begin{equation}
\lambda\Phi_0'(1) \frac{[\tG^L_x(1)]^2}{\tG_x(1,1)} = \tG^E_z(1)\tG^L_x(1).\label{eq:DLM_LoopCount2}
\end{equation}
Here $\tG^E_z(1)=\tG_{xz}(1)/\tG_x(1,1)$ is the probability a node reached by an edge motif is in a loop motif and $\tG^L_x(1)$ is the expected excess degree of a node reached by a loop motif.

\textit{Case 3.} Last, we consider unintended loops that are composed of an edge in a loop motif joined by two edge motifs. Let $u$, $v$ and $w$ be nodes that participate in $k_u$, $k_v$, and $k_w$ edge motifs, respectively. Furthermore, assume $u$ and $v$ are joined by an edge in a loop motif. Then, the probability they form a triangle is
\begin{equation}
	P(u\hspace{-2pt}\sim \hspace{-2pt}w\hspace{-2pt}\sim \hspace{-2pt}v|u\hspace{-2pt}\simL \hspace{-2pt}v) = P(u\sim w)P(w\sim v| u\sim w)\simeq  \frac{k_uk_w}{N\tG_x(1,1)} \frac{(k_w-1)k_v}{N\tG_x(1,1)}.
\label{eq:2path}
\end{equation}
Summing over nodes $w$ and edges $u\simL v$ gives the expected number of triangles
\begin{equation}
	\frac{\lambda N\langle\phi_k\rangle(\sum k \tq_k^L)^2 N (\sum k(k-1)\tp_{k,\ell})}{[N\tG_x(1,1)]^2}
	=\lambda\Phi_0'(1)\frac{[\tG_x^L(1)]^2 \tG_{xx}(1,1)}{[N\tG_x(1,1)]^2}.
\end{equation}
As in Case 2, we simplify this equation by replacing $\lambda \Phi_0'(1)$ with $\tG_z(1)$:
	\begin{equation}
\lambda\Phi_0'(1)\frac{[\tG_x^L(1)]^2 \tG_{xx}(1,1)}{[\tG_x(1,1)]^2} {=} \tG^E_z(1)\tG^L_x(1)\tG^E_x(1,1).\label{eq:DLM_LoopCount3}
\end{equation}
Here, $\tG^E_z(1)$ is the probability a node reached by an edge motif belongs to a loop motif, $\tG^L_x(1)$ is the expected excess degree of a node reached by a loop motif, and $\tG^E_x(1,1)$ is the expected number of excess edge motifs a node reached by an edge motif belongs to.

Thus, the expected number of unintended triangles given by the sum of Eqs.~\eqref{eq:DLM_LoopCount1}, \eqref{eq:DLM_LoopCount2},  and \eqref{eq:DLM_LoopCount3} tends to a finite number independent of $N$ as $N\rightarrow\infty$. Hence, the probability that a randomly picked node participates in an unintended triangle decays as $O(N^{-1})$ as $N\rightarrow\infty$.

\section{Minimum Cycle Basis}
\label{app:MCB}
{
Here we provide a brief overview of relevant topics about cycle spaces. We also verify that the lengths of cycles in a minimum cycle basis (MCB) are unique. For a more detailed review of algorithms for minimum cycle bases see~\cite{kavithaCycle2009}. Many terms here have equivalent definitions in chemistry. A good discussion of these terms is given in~\cite{bergerCounterexamples2004}.

\def\op{\oplus}

Let $G=(V,E)$ be a graph with $|E|=M$ edges and $|V|=N$ vertices. Here a cycle is a walk on the graph $G$ where only the first and last nodes are equal. A cycle of length $\ell$ can be represented by its edges $C = \{(x_1,x_2),(x_2,x_3),\hdots,(x_{\ell},x_1)\}\subseteq E$. The set $C$ can be encoded as a binary incidence vector in $\{0,1\}^M$ where element $i$ designates whether edge $i$ is in $C$. The operation considered on the set of cycles is the symmetric difference defined by $C_1\op C_2=(C_1\setminus C_2)\cup (C_2\setminus C_1)$. In terms of incidence vectors, the symmetric difference is equivalent to addition modulo 2 which can be implemented using the element-wise XOR operation. In this sense, the set of incidence vectors $\{0,1\}^M$ form a vector space over $GF(2)$ (where scalars are also $\{0,1\}$ with addition modulo 2). The set of cycles are not closed under the operation $\op$. See Fig.~\ref{fig:Cyc_Ex}. The cycle space is the vector subspace formed by the closure of cycles under $\op$. Elements of the cycle space are subgraphs where every vertex has even degree. As a vector subspace, the cycle space has a natural notion of basis defined below.

\begin{figure}[!tbp]
	\centering
	\includegraphics{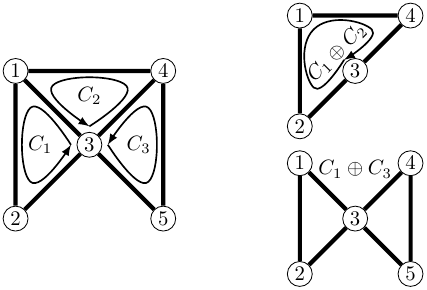}
	\caption{A graph composed of three overlapping cycles $C_1,C_2,C_3$ (left) and depictions of the symmetric differences $C_1\op C_2$ and $C_1\op C_3$ (right). $C_1\op C_2$ forms a cycle while $C_1\op C_3$ does not. The set $\{C_1,C_2,C_3\}$ forms a minimum cycle basis for this graph.}
	\label{fig:Cyc_Ex}
\end{figure}

\newcommand\CS{\mathcal{S}}

\begin{Def}        
	A collection of cycles $B_1,B_2,\hdots,B_k\subset E$ is called a cycle basis if each element $C$ of the cycle space can be uniquely represented as a linear combination
	\begin{equation}
		C = B_{i_1} \op B_{i_2} \op \hdots \op B_{i_\ell},\quad 1\le i_1\le i_2\le \hdots\le i_\ell\le k.
		\label{eq:Cycle_Exp}
	\end{equation}
\end{Def}

A cycle basis can be constructed starting with forming a spanning forest $T$. Consider an edge not contained in our spanning forest. The endpoints of this edge are contained in the same connected component, so there must be a unique path between them in $T$. By combining this edge with the unique path inside $T$ we get a cycle. These cycles are linearly independent as a vector space since they contain an edge that the others do not. A cycle basis constructed in this manner is called a fundamental cycle basis. From this, we find that the circuit rank or the size of the cycle basis is the number of edges not in $T$. This is given by 
\begin{equation}
	M - N + \#CC
\end{equation}
where $\#CC$ is the number of connected components in the graph. More theory on cycle spaces can be found in many texts on graph theory, e.g.~\cite{bollobasModern1998,diestelGraph2017}, or in the review paper~\cite{kavithaCycle2009}.

A graph may have many cycle bases. For instance, $\{C_1,C_2,C_3\}$ and $\{C_1,C_1\op C_2,C_3\}$ are cycle bases of the graph depicted in Fig.~\ref{fig:Cyc_Ex}. However, it is natural to say that a large loop is composed of smaller loops than the other way around. Thus, $\{C_1,C_2,C_3\}$ is a preferable choice of cycle basis in this example. More generally, we want a minimum cycle basis that satisfies the following definition.
\begin{Def}
	A cycle basis $\{B_1,\hdots,B_k\}$ is a minimum cycle basis if it minimizes
	\begin{equation}
		|B_1| + \hdots + |B_k|
	\end{equation}
	where $|B_i|$ is the number of edges in cycle $i$.
\end{Def}

It was first shown in~\cite{hortonPolynomialTime1987} that a minimum cycle basis can be computed in polynomial time. To compute a minimum cycle basis, we use the algorithm from~\cite{kavithaTilde2008} implemented in NetworkX~\cite{hagbergExploring2008}. For unweighted graphs, this algorithm has complexity $O(M^2 N)$. In our experience, the runtime of this algorithm improves by partitioning the graph into biconnected components. We note, in chemistry, the terms the \textit{smallest set of smallest rings} and \textit{minimum cycle basis} are often used interchangeably~\cite{bergerCounterexamples2004}. The smallest set of smallest rings is implemented in many chemistry packages such as the Chemistry Development Kit~\cite{mayEfficient2014,willighagenChemistry2017}.

\begin{figure}[!tbp]
	\centering
	\includegraphics{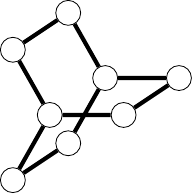}
	\caption{Graph that does not have a unique minimum cycle basis. The cycle space of this graph contains three cycles of length 6, any pair of which forms a minimum cycle basis. This structure is physically feasible and is a building block of the hexagonal diamond lattice}
	\label{fig:MCB_Unique}
\end{figure}

A minimum cycle basis is not necessarily unique. A canonical example of this is the graph depicted in Fig.~\ref{fig:MCB_Unique}. For this reason, minimum cycle bases are often avoided in studying chemical structures. However, the lengths of cycles is the same for every minimum cycle basis. We give a proof of this result in Lemma~\ref{Lem:Unique_Cyc_Len} below. 

In the work~\cite{gleissInterchangeability2000}, a more general result was shown that the set of \textit{relevant cycles} can be partitioned into equivalence classes, where the set of relevant cycles is the union of all minimum cycle bases introduced in~\cite{vismaraUnion1997}. These equivalence classes dictate which cycles can be interchanged for another, and all cycles in an equivalence class are of the same length.

\begin{figure}[!tbp]
	\centering
	\includegraphics[]{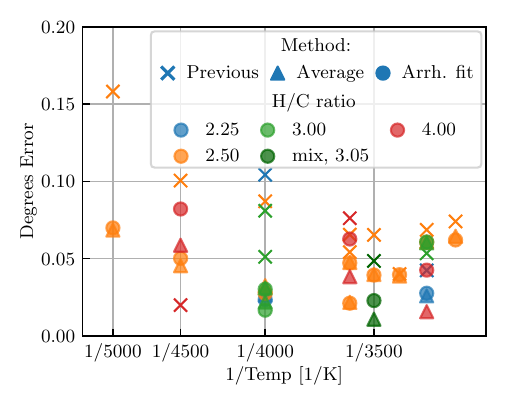}
	\includegraphics{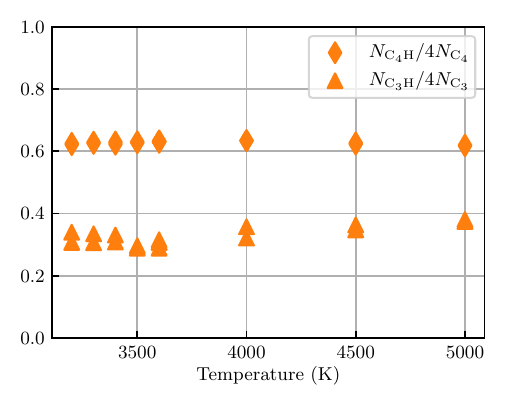}
	\caption{(left) Error given by Eq.~\eqref{eq:DegErr} of the degree distribution as estimated by the ten-reaction model from Previous Work~\cite{dufour-decieuxPredicting2023} as well as from the parameters $\phh$ and $\bp_3$ as introduced in Sec.~\ref{sub:skel_deg}. The parameters $\phh$ and $\bp_3$ are obtained by time-averaging MD data and by the equilibrium constants $\KHHE$ and $\KDoubleE$ fit to an Arrhenius law as discussed in Sec.~\ref{sub:deg_param_fit}. (right) Fraction of bonds starting from carbon atoms of degree 3 or 4 attaches to a hydrogen atom in \ch{C4H10} data. These fractions generalize the bond bias probability $\pch$ by conditioning on total degree. Stubs from carbon atoms bonded to 4 atoms are more likely to bond to hydrogen.}
	\label{fig:Deg_Validation}
\end{figure}

\begin{Lemma}
	Let $\{A_1,\hdots,A_k\},\{B_1,\hdots,B_k\}$ be two minimum cycle bases for a given graph $G$. Assume these cycle bases are given in ascending order, namely $|A_1|\le\hdots\le|A_k|$ and $|B_1|\le\hdots\le |B_k|$. Then, $|A_i|=|B_i|$ for $1\le i\le k$.
	\label{Lem:Unique_Cyc_Len}
\end{Lemma}
\begin{proof}
	Assume towards a contradiction that the vectors $(|A_1|,|A_2|,\hdots,|A_k|)$ and $(|B_1|,|B_2|,\hdots,|B_k|)$ are not equal. Let $j$ be the first index such that $|A_j|\ne |B_j|$, and without loss of generality assume $|A_j|<|B_j|$. The cycles $\{A_1,\hdots,A_j\}$ are linearly independent. Thus, these cycles are not contained in the span of $\{B_1,\hdots,B_{j-1}\}$. In particular, there exists $j_1\le j$ such that 
	\begin{equation}
		A_{j_1} = B_{i_1} \op B_{i_2} \op \hdots \op B_{i_\ell}, \quad i_1\le i_2\le\hdots\le i_\ell \label{eq:Ai_exp}
	\end{equation}
	where $i_1<i_2<\hdots<i_\ell$ and $i_\ell\ge j$. Next, we add
	$B_{i_1}\op\hdots\op B_{i_{\ell-1}}$ to both sides of Eq.~\eqref{eq:Ai_exp} to obtain
	\begin{align}
		A_{j_1} \op B_{i_1} \op \hdots \op B_{i_{\ell-1}} = B_{i_\ell}.
	\end{align}
	Here we used the fact that $C\op C=\emptyset$ for all $C\subset E$.
	Consider the set of cycles 
	\begin{equation}
		\mathcal{B} = \{B_1,\hdots,B_{i_\ell-1},A_{j_1},B_{i_\ell+1},\hdots,B_k\}.
	\end{equation}
	We have shown that $B_{j_2}$ is in the span of $\mathcal{B}$. This implies that $\{B_1,\hdots,B_k\}$ is contained in the span of $\mathcal{B}$, making it a cycle basis. We also have that $j_1\le j \le i_\ell$ so
	\begin{equation}
		|A_{j_1}| \le |A_j| < |B_j| \le |B_{i_\ell}|.
	\end{equation}
	This implies the sum of cycle lengths in $\mathcal{B}$ is smaller than $|B_1|+\hdots+|B_k|$. This contradicts the assumption that $\{B_1,\hdots,B_k\}$ is a minimum cycle basis proving our desired result.
\end{proof}
}

\section{Parameter Estimation}
\label{app:Params}

\subsection{Degree Distribution}
\label{app:Params_Deg}

\subsubsection{Error of the parametric model for the degree distribution}
\label{app:deg_error}

To measure error of the degree distribution we use
\begin{equation}
E(p_k) = \sum_{k=1}^4 \left|p_k - p_k^{\text{(MD)}}\right|
\label{eq:DegErr}
\end{equation}
where $p_k^{\text{(MD)}}$ is the time-averaged fraction of carbon atoms in MD data that are bonded to $k$ carbon atoms.
The error of the degree distribution is shown in Fig.~\ref{fig:Deg_Validation}. We find that the parametric model for the degree distribution introduced in Sec.~\ref{sub:skel_deg} improves upon previous work, especially for high temperature and low H/C ratio. 

We remark that this model does not capture all characteristics of the degree distribution for the global hydrocarbon network $\{\bar{p}_{k_{\rm C},k_{\rm H}}\}$. In particular, Fig.~\ref{fig:Deg_Validation} shows that bonds from degree 4 carbon atoms are more likely to lead to hydrogen than degree 3 carbon atoms. One cause of this is that many degree 3 carbon atoms are the result of double bonds, which is necessarily a carbon-carbon bond. We save these model refinements for future work. 

\subsubsection{Arrhenius fits for \texorpdfstring{$\KHHE$}{Khh} and \texorpdfstring{$\KDoubleE$}{Kc=c}}

\begin{figure}[!tbp]
\centering
\includegraphics{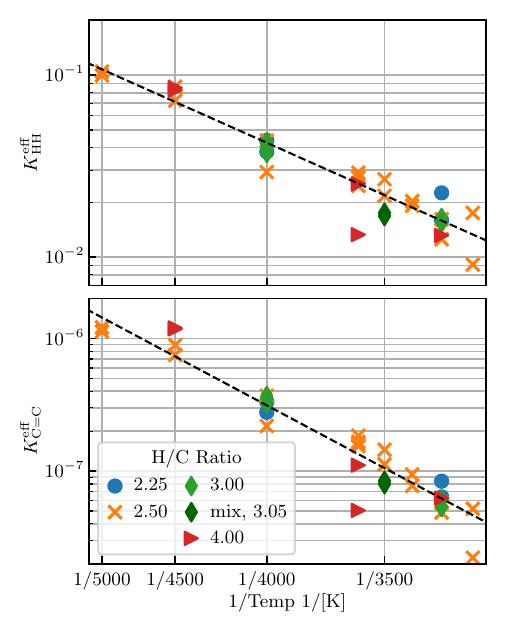}
\caption{Arrhenius fits (dashed, black lines) of the equilibrium constants $\KHHE{=}A_{\rm HH} e^{-C_{\rm HH}/T}$ (top) and $\KDoubleE {=} A_{\rm C=C} e^{-C_{\rm C=C}/T}$ (bottom) compared to estimates from time-averaged MD data (markers), where $A_{\rm HH} {=} 4.4056$, $C_{\rm HH} {=} 18705$, $A_{\rm C=C}{=}6.46{\times} 10^{-4}$, and $C_{\rm C=C} {=} 30526$.}
\label{fig:deg_arrh_fit}
\end{figure}

Recall, the effective equilibrium constants $\KHHE$ and $\KDoubleE$ are given by
\begin{align}
\KHHE &= \frac{\left(4\tfrac{\Nc}{\Nh} - (1-\phh)\right)\phh}{(1-\phh)^2}\\
\KDoubleE &= \frac{\bp_3\Nh\phh}{\bp_4(12.011\Nc+\Nh)T}.
\end{align}
We show the Arrhenius fits for $\KHHE$ and $\KDoubleE$ in Fig.~\ref{fig:deg_arrh_fit}. 
These plots are equivalent to the plots of $\phh$ and $\bp_3$. Specifically, $\KHHE$ and $\KDoubleE$ are estimated from the time-averaged values of $\phh$ and $\bp_3$, and the fit values of $\phh$ and $\bp_3$ are obtained from the Arrhenius fits of $\KHHE$ and $\KDoubleE$. 

\subsection{Loop Rates}
{\label{app:Params_Loop}
In this section, we fit loop counts to MD data. More specifically, we fit the rates at which loops of length $k$ appear per carbon atom, $\lambda \phi_k$. Given the loop rate distribution $\{\lambda\phi_k\}$, the loop rate per carbon atom and loop length distribution are given by $\lambda=\sum_k\phi_k$ and $\phi_k=\lambda\phi_k/\sum_{k'}\lambda\phi_{k'}$, respectively.
The parameters $\lambda$ and $\{\phi_k\}$ can then be used to sample from the Disjoint Loop Model. As outlined previously, we introduce a local reaction for each loop rate $\lambda\phi_k$. For each reaction, we fit the associated effective equilibrium constant to an Arrhenius law. 

\begin{figure}[!tbp]
	\includegraphics{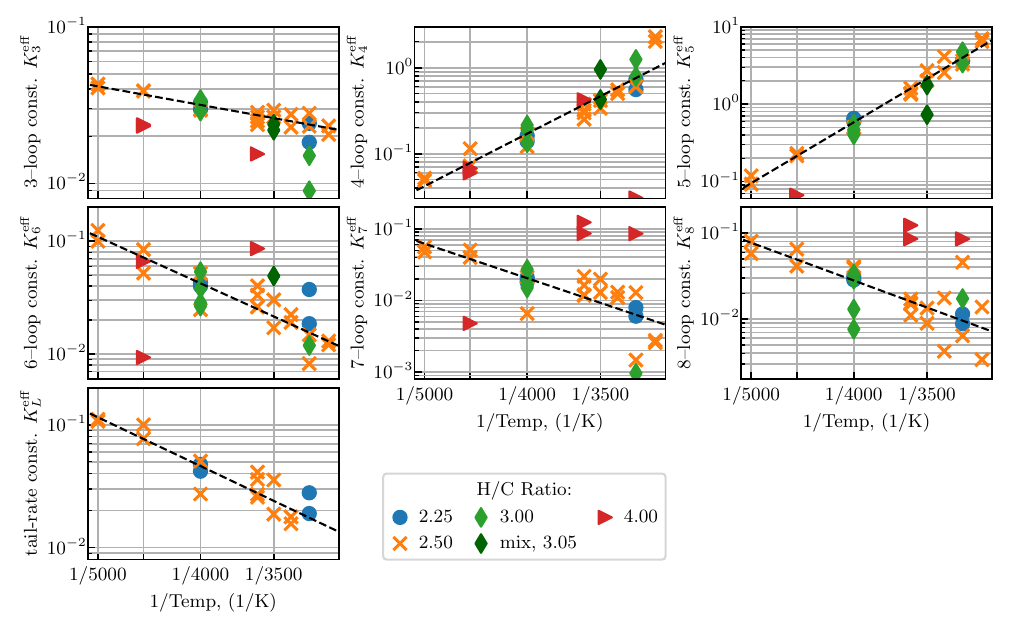}
	\caption{Arrhenius fits for effective equilibrium constants $K_\ell^\text{eff}$, $\ell=3,4,\hdots,8,$ and $K_L^\text{eff}$. For equilibrium constants of the form $K_\ell^\text{eff}$, all \ch{C4H10} data was used as training data. For $K_L^\text{eff}$, the \ch{C4H10} samples at 3200K and 3300K were omitted from training because there existed values $\phi_\ell=0$ for $8<\ell\le 15$.}
	\label{fig:loop_RR}
\end{figure}

We organize our reactions as follows. A set of local reactions is outlined for small cycles, specifically those of length $k=3,4,5$. For loops of length $k=6,7,\hdots$ we use a ring expansion reaction. A ring expansion reaction is a reaction that produces a larger ring from a smaller one. For the reactions discussed in this section, we do not estimate direct counts of structures in molecules. Instead, these reactions are used to develop approximate relationships between the parameters $\phh,\bp_3,\lambda,\{\phi_k\}$. To fully encapsulate all details about ring formation one would need to consider several factors, e.g. the presence of double bonds/radicals, fused rings, and conjugated systems.

We start by considering the formation of rings of length 3 from double bonds
\begin{equation}
	\chemfig{=[1]-[7]-R}\quad\ch{<=>}\quad \chemfig{((-[5]\charge{225=\.}{})-[7]\charge{315=\.}{})-[2]R}\quad\ch{<=>}\quad \chemfig{([:-150]*3(---))-[2]R}
	\label{reac:cyc_prop}
\end{equation}
where R is arbitrary. The intermediate state shows R moving its bond from the right atom to the central atom. This results in a diradical structure that is high energy and short-lived. It is more stable for this diradical to close into a ring or return to a double bond. Reactions of this form have been studied both theoretically~\cite{doubledayDirect1996} and experimentally~\cite{pedersenValidity1994}, although this does not necessarily generalize to the conditions studied here. The number of loops of length 3 is simply $\Nc\lambda\phi_3$. The main component in the reactant is the double bond. Thus, we estimate the number of reactants as the number of double bonds. Assuming most degree 3 atoms are in double bonds, this is approximately $\Nc \bp_3/2$. This motivates the effective equilibrium constant
\begin{equation}
	K_3^\text{eff} = \frac{\lambda \phi_3}{\bp_3}.
	\label{eq:cyc_prop_RR}
\end{equation}
Carbon atoms of degree $3$ may be a consequence of radicals. If they are in a chain of length 3 then they have the potential to close into a cycle. This may contribute to the equilibrium constant $K_3^\text{eff}$ as well.

For loops of length $4$ and $5$ we consider the reactions
\begin{align}
	\chemfig{[:-45]*4(=-[,,,,draw=none]=-)} \quad&\ch{<=>}\quad \chemfig{[:-45]*4(-=--)}\quad,
	\label{reac:cyc_but}\\
	\chemfig{[:-360/10]*5(=-[,,,,draw=none]=--)} \quad&\ch{<=>}\quad
	\chemfig{[:-360/10]*5(-=---)}\quad.
	\label{reac:cyc_pent}
\end{align}
More concretely, loops of length 4 (5) form by closing a carbon chain of length 4 (5) whose endpoints have degree 3. The reactants in Reaction~\eqref{reac:cyc_but} and~\eqref{reac:cyc_pent} are carbon chains with two double bonds. We estimate the number of reactants by $\Nc (\bp_3)^2$ resulting in the effective equilibrium constantss
\begin{align}
	K_4^\text{eff} &= \frac{\lambda \phi_4}{(\bp_3)^2},\label{eq:cyc_but_RR}\\
	K_5^\text{eff} &= \frac{\lambda\phi_5}{(\bp_3)^2}.
\end{align}
Alternatively, a ring of length $4$ can be obtained from two disconnected double bonds, e.g. the formation of cyclobutane from two ethylene molecules. This would result in a similar equilibrium constant to Eq.~\eqref{eq:cyc_but_RR} with an additional term corresponding to volume.

\begin{figure}[!tbp]
	\centering
	\includegraphics{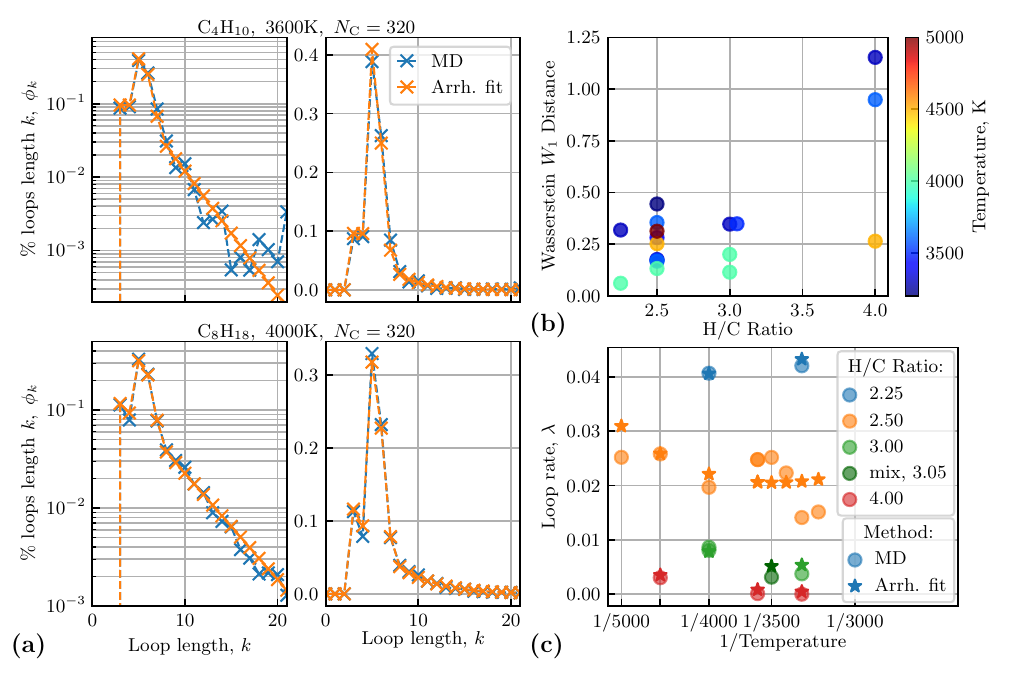}
	
	\caption{(a) Examples of the loop length distribution $\{\phi_k\}$ and its estimated value using the reaction rates $K_\ell^\text{eff}$ and $K_L^\text{eff}$. (b) Wasserstein distance of the estimated loop length distribution obtained by Arrhenius fits to the loop length distribution time-averaged from MD data. (c) Comparison of the loop rate per carbon atom $\lambda$ obtained from Arrhenius fit (stars) vs MD data (circles).}
	\label{fig:loop_reax_results}
\end{figure}

Unlike short chains, the endpoints of long carbon chains are not automatically close together. 
Thus, we consider a reaction that incrementally increases the size of an existing ring:
\begin{equation}
	\chemnameinit{\chemfig{[:-30]*6(------)}}
	\schemestart
	\chemname{\chemfig{[:-54]*5(--[,,,,dashed]---)}}{\text{$\ell$ ring}}
	\+
	\chemname{\chemfig{H-[1]-[7]H}}{}
	\ch{<=>}
	\chemname{\chemfig{[:-30]*6(--[,,,,dashed]----)}}{\text{$\ell+1$ ring}}
	\+
	\chemname{\chemfig{H_2}}{}.
	\schemestop
	\chemnameinit{}
	\label{reac:ring_exp}
\end{equation}
In more detail, one of the hydrogen atoms may break off of the \ch{H-C-H} triple resulting in a radical. This radical can then collide with the loop of length $\ell$ breaking it into a chain of length $\ell+1$. The ends of this chain are close together and one end of the chain is a radical. This chain can then close and break the other hydrogen off of the $\ch{H-C-H}$ reactant. This results in a loop of length $\ell+1$ and 2 hydrogen atoms that are free to bond together.

We may count $\ch{H-C-H}$ triples directly using the degree distribution of the global hydrocarbon network. In particular, we count $H-C-H$ triples by the central carbon atom. If a carbon atom is connected to $k_{\ch{H}}$ hydrogen atoms, then it participates in $\binom{k_{\ch{H}}}{2}$ of these triples. Thus, the rate of \ch{H-C-H} triples per carbon atom is
\begin{equation}
	r_{\ch{HCH}} = \sum_{k_{\ch{C}},k_{\ch{H}}} \binom{k_{\ch{H}}}{2}\bp_{k_{\ch{C}},k_{\ch{H}}} = \frac{1}{2}\frac{\partial^2 \bar{G}}{\partial y^2}(1,1).
\end{equation}
The effective reaction corresponding to Reaction~\eqref{reac:ring_exp} is then
\begin{equation}
	K_\ell^\text{eff} = \frac{\lambda \phi_\ell}{\lambda \phi_{\ell-1}} \frac{\Nh \phh}{\Nc r_{\ch{HCH}}},
\end{equation}
When $\bG$ is given by Eq.~\eqref{eq:DegDistGlobal}, we get that $r_{\ch{HCH}} = 3\pch^2(2-\bp_3)$. We use this value for sampling $r_{\ch{HCH}}$.

For small $\ell$, we fit $K_\ell^\text{eff}$ to an Arrhenius law $A_\ell e^{-C_\ell/T}$.
For large $\ell$, we conjecture that $K_\ell^\text{eff}$ converges to a constant $K_\ell^\text{eff}\to K_L^\text{eff}$ as $\ell\to\infty$. This assumption accounts for noisy data since large loops are rare. In practice, we assume $K_\ell^\text{eff} = K_L^\text{eff} = A_L e^{-C_L/T}$ for $\ell> \ell_{m}$. Then, for $\ell> \ell_{m}$ the ratio $\phi_\ell/\phi_{\ell-1}$ is constant for fixed temperature and H/C ratio. This suggests that the tail of $\{\phi_k\}$ is geometric. In particular we fit $\phi_\ell = \alpha \rho^\ell$ ($\ln \phi_\ell = \ln \alpha + \ell \ln \rho$) over the interval $\ell_m\le \ell \le \ell_M$ using linear least squares. Then, we fit the following to an Arrhenius law
\begin{equation}
	K_L^\text{eff} = \rho \frac{\Nh \phh}{\Nc r_{\ch{HCH}}}.
\end{equation}
The values $\ell_m$ and $\ell_M$ should be chosen such that $K_\ell^\text{eff}$ is roughly constant but $\rho$ is not too noisy.
In our case, we take $\ell_m = 8$ and $\ell_M = 15$.
This choice of geometric tail makes the generating function of $\{\phi_k\}$ a rational function
\begin{equation}
		\Phi_0(x) = \phi_3 x^3 {+} \phi_4 x^4{+}\hdots {+}\phi_{\ell_m}x^{\ell_m} \sum\nolimits_{j=0}^\infty (\rho x)^j
		= \phi_3 x^3 + \phi_4 x^4+\hdots + \frac{\phi_{\ell_m}x^{\ell_m}}{1-\rho x}.
	\label{eq:phi_gen}
\end{equation}
It is necessary that $\rho<1$ for $\{\phi_k\}$ to be a probability mass function. This makes Eq.~\eqref{eq:phi_gen} well defined over the disk $|x|< 1/\rho$.

The Arrhenius fits for the equilibrium constants $K_\ell^\text{eff}$ for $ \ell =3,4,\hdots,8$ and $K_L^\text{eff}$ are shown in Fig.~\ref{fig:loop_RR}. The coefficients for the Arrhenius fits are listed in {Table~\ref{tab:Arrh_Param}.} In our MD data, we find that the distribution $\{\phi_k\}$ is bimodal with peaks at $k=3$ and $k=5$. Loops of length $3$ are easy to form because two neighbors of an atom are automatically physically close. However, short loops experience ring strain (i.e. the bond angles are far from $109.5^\circ$). Thus, while loops of length $5$ and $6$ are larger, they have more stable bond angles. For \ch{C4H10} data loops of length 5 are more common than loops of length 3 except at the highest temperature (5000K). Two sample distributions of $\{\phi_k\}$ are shown in Fig.~\ref{fig:loop_reax_results}(a).

The error of the loop length distribution $\{\phi_k\}$ is depicted in Fig.~\ref{fig:loop_reax_results}(b). For the most part, the prediction of $\{\phi_k\}$ is better for low H/C ratio and high temperature. High temperatures are preferable as they reach equilibrium faster. A low H/C ratio is easier to fit to as loops are more common. In Fig.~\ref{fig:loop_reax_results}(c) we compare the loop rate per carbon atom to MD data.
}

\section{Supplementary Figures}
\label{app:extra_figs}
Here we include additional figures for the random graph models to predict the MD data. In Table~\ref{tab:summary+giant} we list summary statistics for the MD data considered here. In Figs.~\ref{fig:giant_hists_C4,C8} and~\ref{fig:giant_hists_C,C2,mix} we compare the size distribution of the giant connected component from Proposed Model to MD data. In Figs.~\ref{fig:small_dist_C,C2,mix} and~\ref{fig:small_dist_C4,C8} we compare the small molecule size distribution from Proposed Model, Model 2, and Previous Work~\cite{dufour-decieuxPredicting2023} to MD data.

\begin{landscape}
\begin{table*}
	\centering
	\setlength{\tabcolsep}{1.72pt}
	\footnotesize
	\begin{tabular}{|c|c|c|c||c|c|c|c|c|c|c|c|c|c||c|c|c|c|c|}\hline
		\multicolumn{4}{|c||}{Dataset} & \multicolumn{10}{c||}{Summary Data} & \multicolumn{5}{c|}{Largest Molecule, \#C}\\
		\hline
		\multirowcell{2}{Initial\\Composition} & \multirow{2}{*}{Temp.} & \multirow{2}{*}{$\Nc$} & \multirow{2}{*}{$\Nh$} &
		\multicolumn{2}{c|}{$\phh$} & \multicolumn{2}{c|}{$\bp_3$} & \multicolumn{2}{c|}{$\lambda$} & 
		\multicolumn{2}{c|}{$\langle\phi_k\rangle$} & \multicolumn{2}{c||}{$r$ (assort. coeff.)} &
		\multicolumn{2}{c|}{Previous Work~\cite{dufour-decieuxPredicting2023}} & \multicolumn{2}{c|}{Current Work} & \multirowcell{2}{MD}\\
		\cline{5-18} & & & & MD & Fit & MD & Fit & MD & Fit & MD & Fit & Global & Skeleton & Sampled & Gen. Func. & Sampled & Gen. Func. & \\
		\hline\hline
		\ch{C4H10} & 3200K & {256} & 640 & 0.020 & 0.020 & 0.032 & 0.042 & 0.015 & 0.021 & 5.789 & 5.575 & -0.60 & -0.12 & 109.6$\pm$35.3 & 124.0 & 64.3$\pm$26.8 & 34.4 & 45.0$\pm$17.2\\
		\ch{C4H10} & 3300K & {256} & 640 & 0.022 & 0.023 & 0.044 & 0.049 & 0.014 & 0.021 & 5.505 & 5.560 & -0.61 & -0.05 & 110.8$\pm$35.4 & 125.3 & 64.2$\pm$26.6 & 38.3 & 46.3$\pm$13.0\\
		\ch{C4H10} & 3400K & {256} & 640 & 0.030 & 0.027 & 0.054 & 0.056 & 0.022 & 0.021 & 5.839 & 5.543 & -0.60 & -0.04 & 112.0$\pm$35.8 & 126.7 & 65.1$\pm$26.7 & 41.4 & 59.1$\pm$23.8\\
		\ch{C4H10} & 3500K & {256} & 640 & 0.035 & 0.031 & 0.068 & 0.064 & 0.025 & 0.021 & 5.708 & 5.523 & -0.59 & -0.02 & 113.7$\pm$35.7 & 128.4 & 67.2$\pm$27.4 & 43.8 & 64.2$\pm$20.5\\
		\ch{C4H10} & 3600K & {256} & 640 & 0.041 & 0.036 & 0.083 & 0.073 & 0.025 & 0.021 & 5.892 & 5.500 & -0.59 & \phantom{-}0.00 & 115.7$\pm$35.9 & 130.8 & 67.3$\pm$27.5 & 45.6 & 69.6$\pm$22.6\\
		\ch{C4H10} & 3600K & \textbf{320} & 800 & 0.038 & 0.036 & 0.080 & 0.073 & 0.025 & 0.021 & 5.708 & 5.500 & -0.59 & -0.00 & 145.9$\pm$42.9 & 163.5 & 79.2$\pm$33.0 & 57.1 & 71.6$\pm$26.6\\
		\ch{C4H10} & 4000K & {256} & 640 & 0.050 & 0.056 & 0.118 & 0.115 & 0.020 & 0.022 & 5.359 & 5.387 & -0.58 & -0.00 & 127.3$\pm$35.7 & 142.3 & 66.0$\pm$27.1 & 47.2 & 52.1$\pm$19.2\\
		\ch{C4H10} & 4500K & {256} & 640 & 0.094 & 0.084 & 0.186 & 0.184 & 0.026 & 0.026 & 5.042 & 5.220 & -0.53 & \phantom{-}0.01 & 146.5$\pm$33.4 & 159.8 & 64.1$\pm$27.0 & 36.0 & 56.2$\pm$20.4\\
		\ch{C4H10} & 5000K & {256} & 640 & 0.113 & 0.115 & 0.232 & 0.267 & 0.025 & 0.031 & 4.817 & 5.047 & -0.51 & \phantom{-}0.02 & 165.7$\pm$29.2 & 176.2 & 57.3$\pm$24.5 & 9.0 & 51.9$\pm$19.1\\
		\hline
		\ch{C8H18} & 3300K & {288} & 648 & 0.023 & 0.018 & 0.063 & 0.066 & 0.042 & 0.043 & 6.768 & 6.068 & -0.55 & \phantom{-}0.00 & 206.0$\pm$21.9 & 211.8 & 126.6$\pm$39.6 & 141.4 & 135.0$\pm$36.3\\
		\ch{C8H18} & 4000K & \textbf{320} & 720 & 0.045 & 0.045 & 0.146 & 0.149 & 0.041 & 0.040 & 6.266 & 6.009 & -0.54 & -0.01 & 240.8$\pm$20.7 & 245.8 & 134.8$\pm$42.3 & 149.4 & 113.1$\pm$39.3\\
		\hline
		\ch{C2H6} & 3300K & {250} & 750 & 0.039 & 0.038 & 0.023 & 0.026 & 0.004 & 0.005 & 5.040 & 5.038 & -0.69 & -0.01 & 23.6$\pm$10.2 & 0.0 & 19.6$\pm$7.4 & 0.0 & 18.6$\pm$5.7\\
		\ch{C2H6} & 4000K & {250} & 750 & 0.080 & 0.083 & 0.080 & 0.070 & 0.009 & 0.008 & 4.591 & 4.700 & -0.63 & \phantom{-}0.00 & 31.7$\pm$14.4 & 0.0 & 24.2$\pm$9.7 & 0.0 & 21.3$\pm$6.8\\
		\ch{C2H6} & 4000K & \textbf{320} & 960 & 0.083 & 0.083 & 0.079 & 0.070 & 0.008 & 0.008 & 4.500 & 4.700 & -0.63 & \phantom{-}0.01 & 36.0$\pm$16.5 & 0.0 & 26.7$\pm$10.8 & 0.0 & 22.1$\pm$6.6\\
		\hline
		\ch{CH4} & 3600K & {160} & 640 & 0.120 & 0.135 & 0.009 & 0.014 & 0.000 & 0.001 & 3.154 & 4.103 & -0.76 & -0.00 & 5.9$\pm$1.7 & 0.0 & 5.8$\pm$1.5 & 0.0 & 5.4$\pm$1.3\\
		\ch{CH4} & 4500K & {160} & 640 & 0.225 & 0.208 & 0.087 & 0.060 & 0.003 & 0.004 & 3.525 & 3.784 & -0.57 & \phantom{-}0.03 & 9.5$\pm$3.2 & 0.0 & 8.4$\pm$2.6 & 0.0 & 8.9$\pm$2.4\\
		\hline
		\ch{mix} & 3500K & {240} & 732 & 0.044 & 0.052 & 0.031 & 0.034 & 0.003 & 0.005 & 4.553 & 4.901 & -0.69 & -0.00 & 22.4$\pm$9.6 & 0.0 & 18.9$\pm$7.1 & 0.0 & 17.5$\pm$4.9\\
		\hline
	\end{tabular}
	\caption{Summary analysis of ReaxFF MD data and random graph models. Initial conditions are organized by H/C ratio, temperature, and system size (pressure fixed at 40.5 GPa). The parameters for Proposed Model, $\phh$, $\bp_3$, $\lambda$, and $\langle \phi_k\rangle$, are obtained by (1) time-averaging two MD simulations, ``MD'', and (2) equilibrium constants fit to Arrhenius laws as described in Sec.~\ref{Sec:Parameter_Fit}, ``Fit''. The mean loop length $\langle\phi_k\rangle$ is listed to provide information about the loop length distribution $\{\phi_k\}$. The degree assortativity coefficient $r$, Eq.~\eqref{eq:assort_coeff}, is shown for both the global hydrocarbon network and carbon skeleton. The global hydrocarbon network is disassortative by degree and the carbon skeleton lacks assortative mixing by degree. The number of carbon atoms in the largest molecule in MD data is compared to the size of the largest molecule predicted by Previous Work~\cite{dufour-decieuxPredicting2023} and Proposed Model. The mean size of the largest molecule and one standard deviation is given for random graph sampling and MD data. The ten-reaction model from Previous Work~\cite{dufour-decieuxPredicting2023} was sampled 100,000 times and Proposed Model was sampled 50 times over 20,000 iterations of the rewiring algorithm~\ref{alg:rewire} over intervals of 50 rewiring steps at equilibrium. The size of the largest molecule predicted by generating functions is given by $S(1-H(1))$. For Previous Work~\cite{dufour-decieuxPredicting2023} $H=H_0$ is given by the configuration model~\cite{newmanRandom2001} and for Proposed Model $H$ is computed using Eq.~\eqref{eq:HSolveAssort} in Sec.~\ref{subsec:rewire}. If the generating function formalism for Previous Work~\cite{dufour-decieuxPredicting2023} or Proposed Model does not predict a giant connected component then the predicted fraction of nodes in the largest molecule is $S=0$.
	}
	\label{tab:summary+giant}
\end{table*}
\end{landscape}

\begin{figure}[!tbp]
\centering
\includegraphics[width=\textwidth]{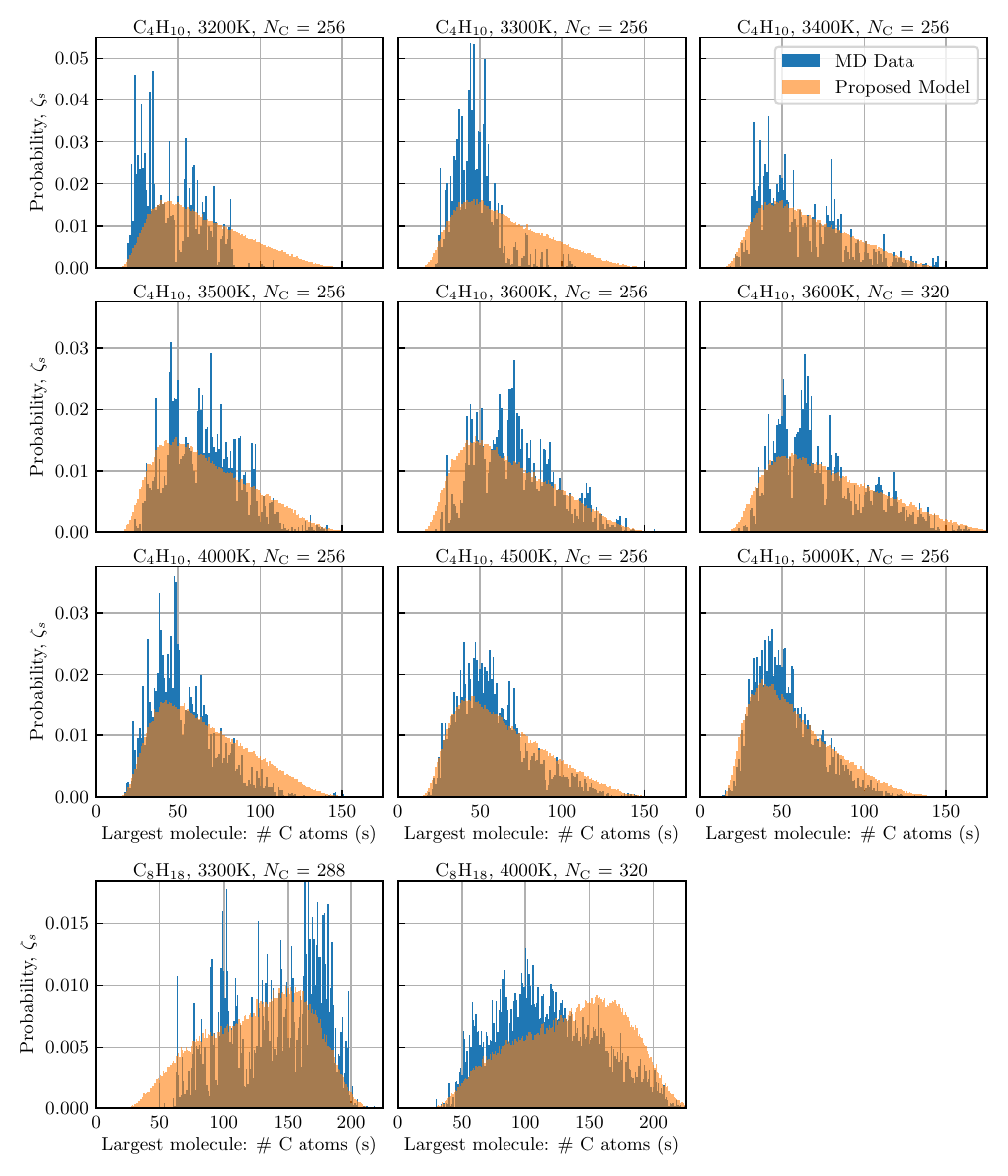}
\caption{Size distribution of the largest molecule from MD data compared to Proposed Model (via random graph sampling) with initial molecular composition \ch{C4H10} and \ch{C8H18}. Under these conditions Proposed Model predicts a giant connected component.}
\label{fig:giant_hists_C4,C8}
\end{figure}

\begin{figure}[!tbp]
\centering
\includegraphics[width=.9\textwidth]{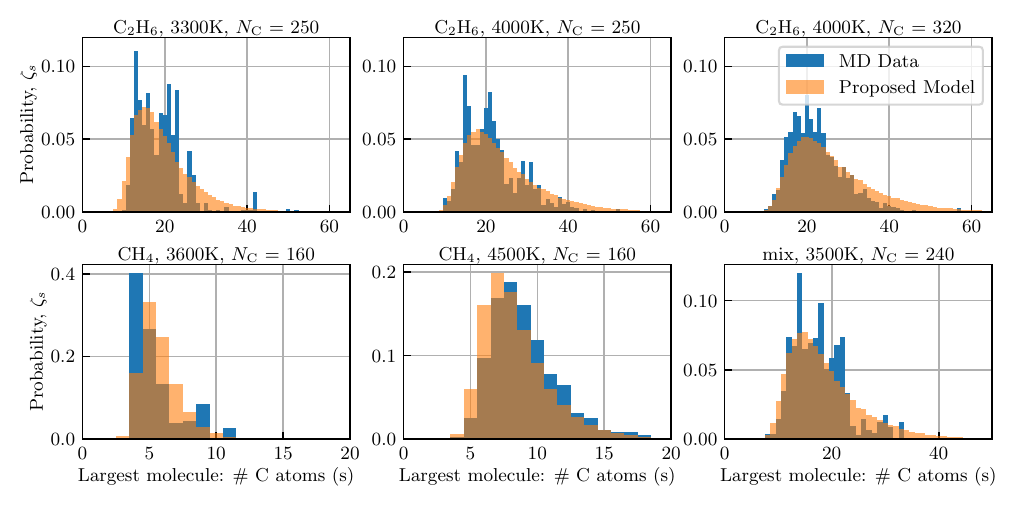}
\caption{Size distribution of the largest molecule from MD data compared to Proposed Model (via random graph sampling) with initial molecular composition \ch{CH4}, \ch{C2H6}, and mix (H/C ratio 3.05). Under these conditions Proposed Model does not predict a giant connected component.}
\label{fig:giant_hists_C,C2,mix}
\end{figure}

\begin{figure}[!tbp]
\centering
\includegraphics[width=.9\textwidth]{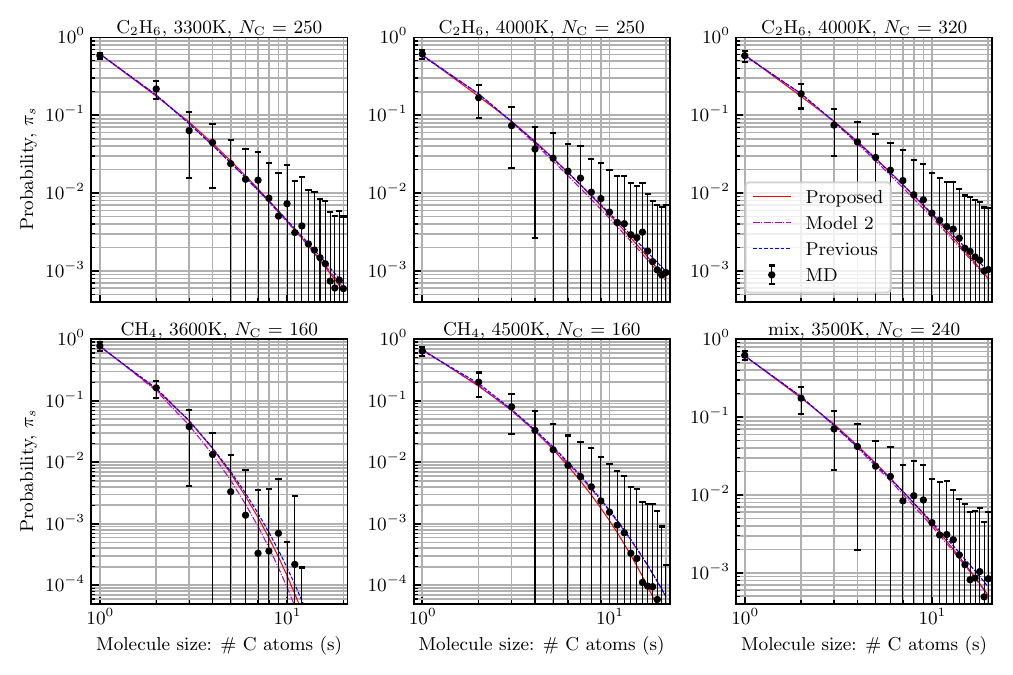}
\caption{Small molecule size distribution from MD data compared to Model 2, Proposed Model with initial molecular composition \ch{CH4}, \ch{C2H6}, and mix (H/C ratio 3.05). The small molecule size distribution for all random graph models are obtained via generating functions.}
\label{fig:small_dist_C,C2,mix}
\end{figure}

\begin{figure}[!tbp]
\centering
\includegraphics[width=\textwidth]{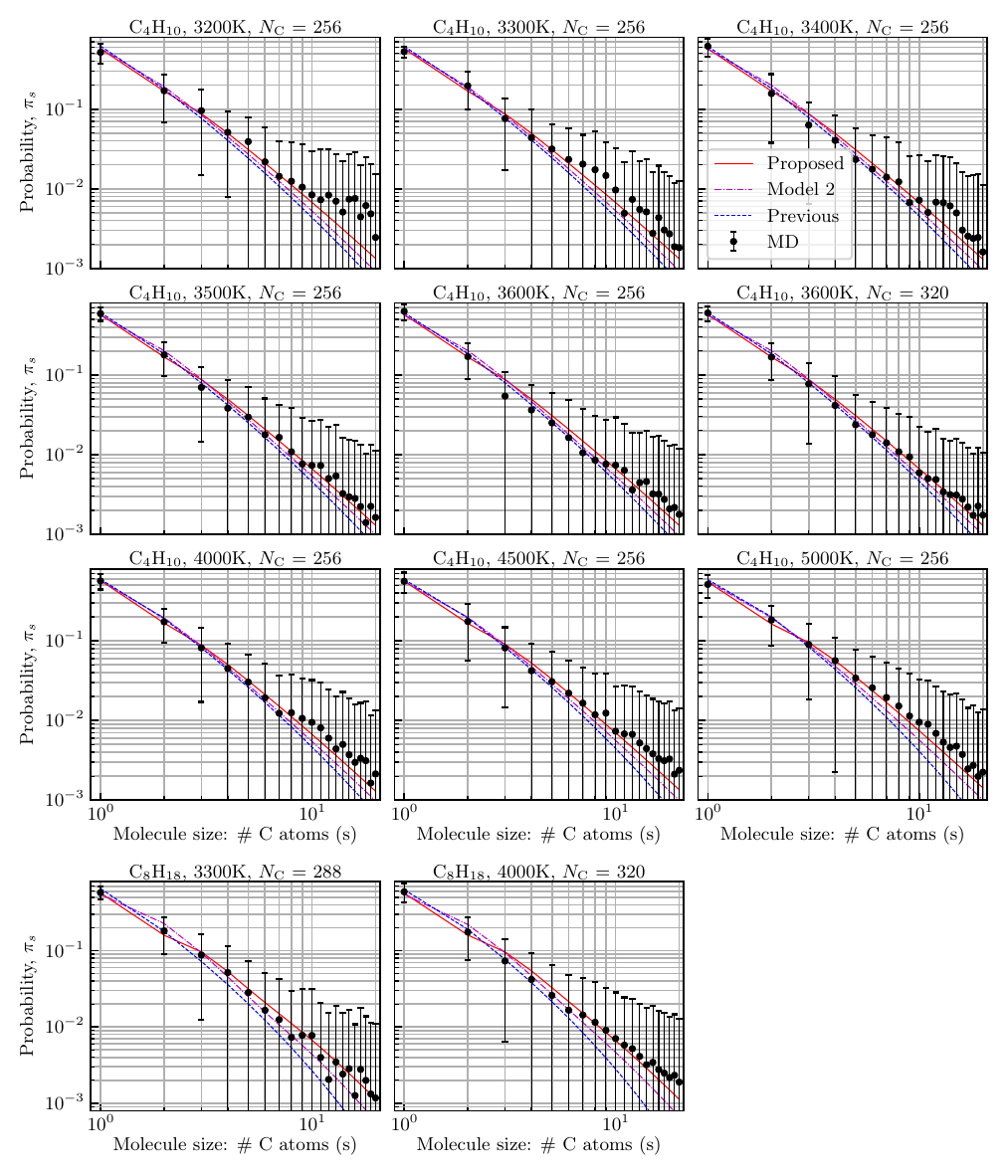}
\caption{Small molecule size distribution from MD data compared to Model 2, Proposed Model with initial molecular composition \ch{C4H10} and \ch{C8H18}. The small molecule size distribution for all random graph models are obtained via generating functions.}
\label{fig:small_dist_C4,C8}
\end{figure}

\end{document}